\documentclass[journal, onecolumn,12pt,draftclsnofoot]{IEEEtran}
\usepackage{amsmath}
\usepackage{mathrsfs}

\usepackage{amsfonts}
\usepackage{cite}
\usepackage{graphicx,color,overpic}
\usepackage{amsmath}
\usepackage{times}
\usepackage{latexsym}
\usepackage{bm}
\usepackage{amssymb}
\usepackage{cases}
\usepackage{array}
\usepackage{fancyhdr}
\usepackage{setspace}
\usepackage{subfigure}
\usepackage{url}
\usepackage{algorithm}
\usepackage{algorithmic}
\usepackage{multirow}

\usepackage{citesort}
\usepackage{epsfig}
\usepackage{fancybox}
\usepackage{textcomp}
\usepackage{multirow}
\usepackage{setspace}
\usepackage{psfrag}

\newtheorem{lemma}{Lemma}

\newtheorem{corollary}{Corollary}
\newtheorem{proposition}{Proposition}

    \def\Complex{{\rm\rule[.23ex]{.03em}{1.1ex}\kern-.3em{C}}}

    \newcommand{\be}{\begin{equation}} \newcommand{\ee}{\end{equation}}
    \newcommand{\bea}{\begin{eqnarray}} \newcommand{\eea}{\end{eqnarray}}
    \newcommand{\benum}{\begin{enumerate}} \newcommand{\eenum}{\end{enumerate}}

    \newcommand{\qb}{{\bf b}}

    \newcommand{\qp}{{\bf p}}

    \newcommand{\qs}{{\bf s}}

    \newcommand{\qv}{{\bf v}}

    \newcommand{\qy}{{\bf y}}
    \newcommand{\qz}{{\bf z}}

    \newcommand{\qA}{{\bf A}}
    \newcommand{\qB}{{\bf B}}
    \newcommand{\qC}{{\bf C}}
    \newcommand{\qD}{{\bf D}}
    \newcommand{\qE}{{\bf E}}
    \newcommand{\qF}{{\bf F}}
    \newcommand{\qG}{{\bf G}}
    \newcommand{\qH}{{\bf H}}
    \newcommand{\qI}{{\bf I}}

    \newcommand{\qQ}{{\bf Q}}
    \newcommand{\qR}{{\bf R}}
    \newcommand{\qS}{{\bf S}}
    \newcommand{\qT}{{\bf T}}
    \newcommand{\qU}{{\bf U}}
    \newcommand{\qV}{{\bf V}}
    
    \newcommand{\qX}{{\bf X}}

    \newcommand{\qzero}{{\bf 0}}
    \newcommand{\qone}{{\bf 1}}

    \newcommand{\qPsi}{{\boldsymbol \Psi}}
    \newcommand{\qPhi}{{\boldsymbol \Phi}}
    \newcommand{\qXi}{{\boldsymbol \Xi}}
    \newcommand{\qTheta}{{\boldsymbol \Theta}}
    \newcommand{\qLambda}{{\boldsymbol \Lambda}}
    \newcommand{\qSigma}{{\boldsymbol \Sigma}}
    
    \newcommand{\qUpsilon}{{\boldsymbol \Upsilon}}
    \newcommand{\qOmega}{{\boldsymbol \Omega}}
    \newcommand{\qPi}{{\boldsymbol \Pi}}

    \newcommand{\qvartheta}{{\boldsymbol \vartheta}}

    \newcommand{\ta}{{\tilde{a}}}
    \newcommand{\tb}{{\tilde{b}}}

    \newcommand{\te}{{\tilde{e}}}

    \newcommand{\tqC}{{\tilde{\qC}}}
    \newcommand{\tqH}{{\tilde{\qH}}}

    \newcommand{\bR}{{\bar{R}}}
    \newcommand{\bI}{{\bar{I}}}

    \newcommand{\bqH}{{\bar{\qH}}}
    
    \newcommand{\bqV}{{\bar{\qV}}}

    \newcommand{\bbR}{{\mathbb R}}
    \newcommand{\bbC}{{\mathbb C}}

    \newcommand{\calG}{{\mathcal G}}
    \newcommand{\calH}{{\mathcal H}}

    \newcommand{\calO}{{\mathcal O}}

    \newcommand{\calR}{{\mathcal R}}

    \newcommand{\calqH}{\boldsymbol{\calH}}

    \newcommand{\diag}{{\sf diag}}
    
    \newcommand{\tr}{{\sf tr}}
    \newcommand{\Ex}{{\sf E}}

    \newcommand{\Ve}{{\sf vec}}

\begin{document}

\title{Large System Achievable Rate Analysis of RIS-Assisted MIMO Wireless Communication \\ with Statistical CSIT}

\author{Jun Zhang$\stackrel{\ast}{,}$\footnote{$^{\ast}$J. Zhang and J. Liu are with Jiangsu Key Laboratory of Wireless Communications, Nanjing University of Posts and Telecommunications, Nanjing 210003, China, Email: \{zhangjun, 1218012037\}@njupt.edu.cn.}  
~Jie Liu$\stackrel{\ast}{,}$
~Shaodan Ma$\stackrel{\dagger}{,}$\footnote{$^{\dagger}$S. Ma is with the State Key Laboratory of Internet of Things for Smart City and Department of Electrical and Computer Engineering, University of Macau, Macao S.A.R. 999078, China, Email: shaodanma@um.edu.mo.
}
~Chao-Kai Wen$\stackrel{\ddagger}{,}$\footnote{$^{\ddagger}$C.-K. Wen is with Institute of Communications Engineering, National Sun Yat-sen University, Kaohsiung 804, Taiwan. Email: chaokai.wen@mail.asysn.edu.tw.}
and~Shi Jin$\stackrel{\S}{}$\footnote{$^{\S}$S. Jin is with the National Mobile Communications Research Laboratory, Southeast University, Nanjing 210096, China, Email: jinshi@seu.edu.cn.}
}

\date{}

\maketitle

\begin{abstract}
Reconfigurable intelligent surface (RIS) is an emerging technology to enhance wireless communication in terms of energy cost and system performance by equipping a considerable quantity of nearly passive reflecting elements. This study focuses on a downlink RIS-assisted multiple-input multiple-output (MIMO) wireless communication system that comprises three communication links of Rician channel, including base station (BS) to RIS, RIS to user, and BS to user. The objective is to design an optimal transmit covariance matrix at BS and diagonal phase-shifting matrix at RIS to maximize the achievable ergodic rate by exploiting the statistical channel state information at BS. Therefore, a large-system approximation of the achievable ergodic rate is derived using the replica method in large dimension random matrix theory. This large-system approximation enables the identification of asymptotic-optimal transmit covariance and diagonal phase-shifting matrices using an alternating optimization algorithm. Simulation results show that the large-system results are consistent with the achievable ergodic rate calculated by Monte Carlo averaging. The results verify that the proposed algorithm can significantly enhance the RIS-assisted MIMO system performance.
\end{abstract}



\section{Introduction}

The fifth-generation (5G) wireless network is being commercially deployed in many countries this year. Compared with the former fourth-generation wireless network, 5G has considerable improvements in many aspects, such as capacity, coverage, privacy, security, user experience, and information interaction.
Meanwhile, supported by its low latency, large bandwidth, and high reliability, 5G introduces possibilities for the implementation of many new technologies, such as extended reality (XR) services, safe and reliable autonomous driving technology, and telemedicine remote control \cite{20Network-Saad}.
A foreseeable future trend is a continuous increase in network users and the high requirements of new applications for data rate, transmission delay, and service reliability. These trends are also the goals of sixth-generation wireless communication in the future.

Many emerging technologies applied in wireless communication systems, such as massive multiple-input multiple-output (MIMO), millimeter-wave communication, ultra-dense networks, have been proposed in recent years to achieve the above goals \cite{14JSAC-Andrews,16JSAC-Buzzi}. These technologies not only effectively improve the spectral efficiency and energy efficiency but also serve a massive number of users. However, these technologies introduce new issues of high energy consumption and increased hardware cost simultaneously \cite{20ComM-Wu,20WC-Huang}.
With the continuously increasing number of users and the additional demand of data rate, the rapidly growing energy consumption and hardware costs required for system operation will be difficult.
Considering that sacrificing wireless resources in exchange for system performance is a short-term solution, new technologies must be developed to improve the energy efficiency and reduce operating costs, such as integrated frequency bands, edge artificial intelligence, integrated terrestrial, satellite networks, and reconfigurable intelligent surface (RIS), continuously \cite{20Network-Saad}.

In particular, RIS has been recently proposed as a new emerging technology to achieve green communication since it can improve the wireless propagation environment by controlling the reflection coefficients \cite{18TSP-Hu,19TWC-Huang,20ComM-Wu,20WC-Huang,20TWC-Tang}.
The traditional reflecting surface was not considered earlier in terrestrial wireless communication systems because it only had fixed reflection phases that cannot be flexibly adjusted according to the signal and cannot be used in a real-time changing communication environment \cite{20TWC-Nadeem}.
Fortunately, recent studies on micro-electrical-mechanical systems and micromaterials have realized new achievements, which can facilitate the reflection phase change with the signal in real-time \cite{14LSA-CuiCoding}. The RIS comprises of many passive reflection units, wherein each reflection unit can independently replicate the incident signals and change their phase or amplitude. Compared with other technologies, RIS has the following advantages. First, RIS comprises many passive reflection units, thus, its operation only requires small energy consumption. Second, the RIS can be flexibly installed in suitable locations, such as the exterior walls of buildings, surfaces of trees or cars, and indoor walls, due to its light weight. Last, a low-cost RIS can obtain antenna gain by improving the wireless propagation environment without additional power consumption.

Motivated by these attractive characteristics of RIS, many works have recently focused on RIS-assisted communication systems, such as RIS-assisted multiple-input single-output (MISO) channels \cite{19TWC-Huang,20TWC-Nadeem,20TWC-Jung,18GLOBECOM-Wu,20WCL-Yan,20TCom-Abeywickrama,20TSP-Zhou} and RIS-assisted MIMO channels \cite{20JSAC-Zhangshuowen,20TWC-PanCH,20JSAC-Pan}. Specifically, in \cite{19TWC-Huang}, the authors considered the energy efficiency optimization problem for RIS-assisted MISO downlink multi-user system with zero-forcing (ZF) precoding to find the optimal RIS phase shifts and the power allocation with the transmit power and quality of service constraints.
An alternating optimization algorithm is proposed to fully reap the beamforming gains of the transmit precoding and RIS for the RIS-assisted MISO downlink system in \cite{18GLOBECOM-Wu} to maximize the received signal power under the transmit power constraint.
In \cite{20TSP-Zhou}, the RIS can enhance the sum rate performance of multi-group multicast MISO communication networks as it can improve the channel condition of the worst-case user in each group.
The RIS-assisted MIMO system was considered in \cite{20JSAC-Zhangshuowen}, in which the RIS reflection coefficients and the transmit covariance matrix were jointly designed.
In \cite{20TWC-PanCH}, the authors further studied the multi-cell MIMO downlink system with inter-cell interference and alternately optimized the precoding matrices and phase shifts.
In \cite{20JSAC-Pan}, the authors showed that the RIS can enhance the operation range of the wireless powered sensors, and at the same time improving the data rate performance of the information receivers in the RIS-assisted MIMO system.

It is worth noting that the aforementioned works are all based on the instantaneous channel state information at transmitter (CSIT). However, knowledge of the statistical CSIT only by the transmitter or RIS is realistic in practical applications because the channel estimation at RIS is difficult \cite{20TWC-Nadeem,20TSP2-Zhou,20WCL-Zhou,20JSAC-Zhangshuowen}.
Under the assumption of imperfect CSIT, the robust beamforming was designed for a RIS-aided multiuser MISO system in \cite{20TSP2-Zhou,20WCL-Zhou}.
By exploiting the statistical CSIT, \cite{19TVT-Han} studied the phase-shift optimization problem of the RIS-assisted MISO downlink single-user system via the maximum ratio combining precoding. Using the optimal linear precoding, \cite{20TWC-Nadeem} studied the max-min signal-to-interference-plus-noise ratio problem of RIS-assisted MISO downlink multi-user systems via random matrix theory, in which the designed phase-shifting matrix only depends on statistical CSIT. In \cite{ZhangICC20}, the optimal transmit covariance matrix and the diagonal phase-shifting matrix are designed by exploiting the statistical CSIT in RIS-assisted MIMO system with Rayleigh channel and without direct link from base station (BS) to user.

This study focuses on a downlink RIS-assisted MIMO wireless communication system comprising a BS, a user, and a RIS. As shown in Fig. \ref{fig:1}, all components are equipped with multiple antennas or nearly passive and cheap reflecting elements. There are three communication links between BS-user, BS-RIS, and RIS-user is assumed, and all links contain line-of-sight (LoS) and non-LoS (NLoS) components. The optimal transmit covariance and the diagonal phase-shifting matrices are obtained by exploiting the statistical CSIT to maximize the achievable ergodic rate. The large-system analysis presented in this study is based on the {\em replica method}. This approach was originally developed in statistical physics \cite{75JPF-Edwards} and successfully applied to wireless communication systems, such as code-division multiple-access channels \cite{02TIT-Tanaka}, MIMO channels \cite{03TIT-Moustakas,08JSAC-Muller,08TIT-Taricco,10TWC-Wen}, and MIMO relay channels \cite{11TCom-Wen,12TWC-Wen}.
In \cite{12TWC-Wen}, the authors analyzed the asymptotic mutual information of the MIMO relay Rician channel with a direct link from source to destination. Although there exist three Rician random matrices, only the product of two Rician random matrices in the large-system limit needs to handle in each time slot since there exist two time slots for the MIMO relay Rician channel. However, the RIS-assisted MIMO system is significantly different from the MIMO relay system in terms of generating new signals in the second time slot. In the RIS-assisted MIMO system with a direct link from BS to user, the {\em product and sum of three Rician random matrices} must be processed simultaneously since the transmission is finished in one time slot. Tackling this challenge makes that the result of this study is non-trivial and novel.
The main contributions of this study are summarized below.
\begin{itemize}
  \item A large-system approximation of the achievable ergodic rate is derived by using {\em replica method} in large dimension random matrix theory. Since the effective channel consists of {\em the Rician random matrix products and sums}, the result is general and can be applied to several scenarios because the effective channel comprises the Rician random matrix products and sums. Simulation results verify that the derived large-system approximation can provide accurate results even for small antenna systems.
  \item The large-system approximation is applied to design the transmit covariance matrix to maximize the achievable ergodic rate of the downlink RIS-assisted MIMO system by employing an iterative water-filling optimization algorithm based on the statistical CSIT. The result can be degraded in some special scenarios.
  \item The design of the diagonal phase-shifting matrix is also obtained by using the projected gradient ascent method. An alternating optimization algorithm is also introduced to find the two aforementioned matrices.
\end{itemize}

The rest of this paper is organized as follows. Section II introduces the channel model and problem formulation. Section III presents the main results, in which a large-system approximation of the achievable ergodic rate will be derived. These results will be used to design the optimal transmit covariance and diagonal phase-shifting matrices. Section IV presents the simulation results and Section V concludes the paper.

\emph{Notations}---We use uppercase and lowercase boldface letters to denote matrices and vectors, respectively. In addition, ${\bf{I}}_N$ denotes an $N \times N$ identity matrix while an all-zero matrix is denoted by ${\bf{0}}$, and an all-one matrix is denoted by ${\bf{1}}$. The matrix inequality $\succeq$ shows the positive semi-definiteness. The superscripts $(\cdot)^{H}$, $(\cdot)^{T}$, and $(\cdot)^{*}$ represent the conjugate-transpose, transpose, and conjugate operations, respectively. Moreover, we use $\Ex\{\cdot\}$ to denote expectation with respect to all random variables within the brackets and $\log(\cdot)$ is the natural logarithm. The complex number field is denoted by $\mathbb{C}$. For any matrix $\qA \in \mathbb{C}^{N \times n}$, we use $[{\bf A}]_{l,k}$ to denote the ($l$,$k$)-th entry, and $a_k$ denotes the $k$-th entry of the column vector $\bf{a}$. The operators $(\cdot)^{\frac{1}{2}}$, $(\cdot)^{-1}$, $(\cdot)^{-H}$, ${\tr}(\cdot)$, $\Ve\{\cdot\}$, and $\det(\cdot)$ represent the matrix principal square root, inverse, inverse and conjugate operations, trace, vectorization, and determinant, respectively. In addition, $\diag(\bf{x})$ denotes a diagonal matrix with an input vector $\bf{x}$ representing its diagonal elements.


\section{System Model and Problem Formulation}

\subsection{System Model}

\begin{figure}
\begin{center}
\resizebox{3.5in}{!}{%
\includegraphics*{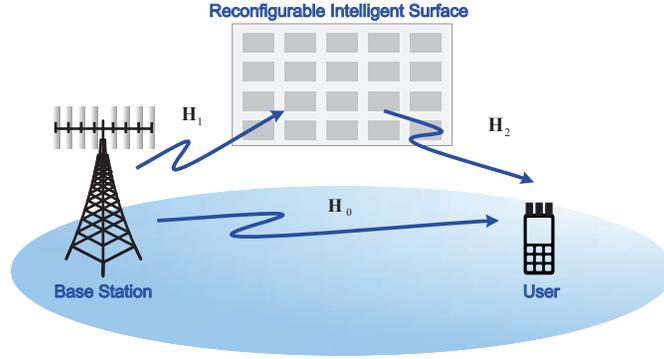} }%
\caption{A downlink RIS-assisted MIMO wireless communication system, which only the statistical CSIT is available at BS.}\label{fig:1}
\end{center}
\end{figure}

As shown in Fig. \ref{fig:1}, we consider a downlink RIS-assisted MIMO wireless communication system comprising a BS equipped with $N \geq 1$ antennas, a user equipped with $K \geq 1$ antennas, and a RIS equipped with $L \geq 1$ nearly passive reflecting elements. It is assumed that $\qH_0 \in \bbC^{K \times N}$, $\qH_1 \in \bbC^{L \times N}$, and $\qH_2 \in \bbC^{K \times L}$ denote the block fading channel matrices of the channels from BS to user, from BS to RIS, and from RIS to user, respectively. The received signals $\qy \in \bbC^K$ at user can be expressed as
\begin{align}\label{eq:the received signal}
 \qy &= (\qH_0 + \qH_2 \qTheta \qH_1) \qs + \qz,
\end{align}
where $\qs \in \bbC^{N}$ denotes the zero-mean transmitted Gaussian vector with covariance matrix $\qQ \in \bbC^{N \times N}$, $\qTheta = \diag\{\xi_1 e^{j\theta_1}, \xi_2 e^{j\theta_2}, \ldots, \xi_L e^{j\theta_L}\}$ represents the diagonal phase-shifting matrix of RIS,
$\theta_l \in [0,2\pi)$ and $\xi_l \in [0,1]$ denote the phase shift and amplitude reflection coefficient of the $l$-th reflecting element, respectively,
and $\qz\in \bbC^K$ is the noise vector whose entries consist of independent zero-mean circularly symmetric complex Gaussian with variance $\sigma^2$. Without loss of generality, we set $\xi_l=1$ for $l=1,2,\ldots,L$. If $\xi_l=0$ for $l=1,2,\ldots,L$, it is the case without RIS. Therefore, the transmit power constraint at BS can be expressed as
\begin{equation}\label{eq:BS transmitted power constrain}
  \tr \Ex\{\qs\qs^H\} = \tr \qQ \leq {P},
\end{equation}
where $P > 0$ is determined by the power budget of BS. 

We use the Kronecker model to characterize the spatial correlation of the MIMO channel for each link. The adopted channel model allows different transmit correlation matrices and LoS components. Specifically, we can write
\begin{align}
 \qH_i  &= \tqH_i + \bqH_i = \qR_i^\frac{1}{2} \qX_i \qT_i^\frac{1}{2} + \bqH_i, \mbox{~for~} i = 0,1,2,\label{eq:channel model}
\end{align}
where $\qR_0 \in\bbC^{K \times K}$, $\qT_0 \in\bbC^{N \times N}$, $\qR_1 \in\bbC^{L \times L}$, $\qT_1 \in\bbC^{N \times N}$, $\qR_2 \in\bbC^{K \times K}$, and $\qT_2 \in\bbC^{L \times L}$ are deterministic nonnegative definite matrices that characterize the spatial correlations of the downlink channel at BS, RIS, and user, respectively, $\qX_0 \equiv [\frac{1}{\sqrt{N}} X_{0,ij}] \in \bbC^{K \times N}$, $\qX_1 \equiv [\frac{1}{\sqrt{N}} X_{1,ij}] \in \bbC^{L \times N}$, and $\qX_2 \equiv [\frac{1}{\sqrt{L}} X_{2,ij}] \in \bbC^{K \times L}$ consist of random components of the three channels in which the elements $X_{0,ij}$'s, $X_{1,ij}$'s, and $X_{2,ij}$'s are independent and identically distributed (i.i.d.)~complex zero-mean random variables with unit variance, and $\bqH_0 \in \bbC^{K \times N}$, $\bqH_1 \in \bbC^{L \times N}$, and $\bqH_2 \in \bbC^{K \times L}$ are deterministic matrices corresponding to the LoS components of the three channels, respectively.

For the above channel models, we define the Rician factors of three channels as
\begin{equation}
 \kappa_i = \frac{\|\bqH_i\|^2_{F}}{\Ex \{\|\tqH_i\|^2_{F}\}}, \mbox{~for~} i = 0,1,2.
\end{equation}
We also denote the large-scale fading coefficients of the three links by $\Gamma_i$ for $i = 0,1,2$. For conciseness, the effects of $\Gamma_i$ are absorbed into $\qT_i$ and $\bqH_i$ for $i = 0,1,2$, respectively. Since the variances of the random component of channels $\qX_0$, $\qX_1$, and $\qX_2$ are $1/N$, $1/N$, and $1/L$, respectively, in \eqref{eq:channel model}, $\qR_i$, $\qT_i$, and $\bqH_i$ (for $i = 0,1,2$) are normalized as follows
\begin{subequations}\label{eq: normalized channel}
\begin{align}
   & \tr \qR_0 = K,  \quad
  \tr \qT_0 = \frac{1}{\kappa_0+1} N^2 \Gamma_0, \quad \tr \bqH_0 \bqH_0^H = \frac{\kappa_0}{\kappa_0+1} NK \Gamma_0, \label{eq: normalized channel 0}\\
  & \tr \qR_1 = L,  \quad
  \tr \qT_1 = \frac{1}{\kappa_1+1} N^2 \Gamma_1, \quad \tr \bqH_1 \bqH_1^H = \frac{\kappa_1}{\kappa_1+1} NL \Gamma_1, \label{eq: normalized channel 1}\\
  &\tr \qR_2 = K,  \quad
  \tr \qT_2 = \frac{1}{\kappa_2+1} L^2 \Gamma_2,  \quad \tr \bqH_2 \bqH_2^H = \frac{\kappa_2}{\kappa_2+1} KL \Gamma_2. \label{eq: normalized channel 2}
\end{align}
\end{subequations}

We further assume that only the statistical CSIT, i.e., $\{\qR_i, \qT_i, \bqH_i, \kappa_i, \Gamma_i, \mbox{~for~} i = 0,1,2\}$, is available at BS since it is more realistic than instantaneous CSI. As such, the achievable ergodic rate $R$ for the MIMO channel can be expressed as
\begin{align}\label{eq:ergodic rate}
 R\left(\qQ,\qTheta \right) = & \Ex_{\{\qH_i,i = 0,1,2\}} \left\{ \log\det \left(\qI_K + \frac{1}{\sigma^2} \qH \qQ \qH^H \right) \right\},
\end{align}
where $\qH = \qH_0 + \qH_2 \qTheta \qH_1$ denotes the effective channel.

\subsection{Problem Formulation}
Our objective is to maximize the achievable ergodic rate subject to the transmit power constraint \eqref{eq:BS transmitted power constrain} at BS by determining the optimal transmit covariance matrix at BS and the optimal diagonal phase-shifting matrix at RIS. Thus, our optimization problem is then
\begin{align}\label{eq: ergodic rate maximization P1}
  ({\rm P}1)~~\max_{\qQ,\qTheta} &~~  R\left(\qQ,\qTheta\right)  \\
  \mbox{s.t.} &~~ \tr \qQ \leq N P, ~~\qQ \succeq \qzero,  ~~ \qTheta = \diag\{e^{j\theta_1}, e^{j\theta_2}, \ldots, e^{j\theta_L}\},~~ \theta_l\in [0,2\pi). \nonumber
\end{align}
However, the quest of the optimal solution for (${\rm P}1$) in \eqref{eq: ergodic rate maximization P1} is extremely challenging because the problem is not a convex problem since the norm of elements in the diagonal phase-shifting matrix $\qTheta$ is 1. In addition, Monte-Carlo averaging over the channels is required to evaluate the achievable ergodic rate $R\left(\qQ,\qTheta \right)$ in \eqref{eq:ergodic rate}, thus making the overall computational complexity prohibitive. To tackle these challenges, we present an approach to solve the optimization problem (${\rm P}1$) in the next section using the large-system regime.

\section{Transmit Covariance Matrix and Phase-Shifting Matrix Optimization}

\subsection{Large System Analysis}

In this section, we shall derive the analytical expression for the ergodic rate in the large-system regime, i.e., $N$, $L$, and $K$ all go to infinity with the ratios $L/N$ and $K/L$ kept constant at $\epsilon_1$ and $\epsilon_2$, respectively. To simplify the notation in the derivation, the effects of $\qQ$ and $\qTheta$ have been incorporated into $\{\qT_i, \bqH_i, i=0,1,2\}$ by the following replacements
\begin{subequations}\label{eq:replacement}
\begin{align}
  & \qT_0:= \qQ^{\frac{1}{2}} \qT_0 \qQ^{\frac{1}{2}}~~\mbox{and} ~~\bqH_0:= \bqH_0 \qQ^{\frac{1}{2}}, \\
  & \qT_1:= \qQ^{\frac{1}{2}} \qT_1 \qQ^{\frac{1}{2}}~~\mbox{and} ~~\bqH_1:= \bqH_1 \qQ^{\frac{1}{2}}, \\
  & \qT_2:= \qTheta^H \qT_2 \qTheta ~~\mbox{and}~~\bqH_2:= \bqH_2 \qTheta.
\end{align}
\end{subequations}
Thus, the achievable ergodic rate $R$ in \eqref{eq:ergodic rate} can be rewritten as
\begin{align}\label{eq:ergodic rate replacement}
 I = & \Ex_{\{\qH_i,i = 0,1,2\}} \left\{ \log\det \left(\qI_K + \frac{1}{\sigma^2} (\qH_0 + \qH_2 \qH_1) (\qH_0 + \qH_2 \qH_1)^H \right) \right\}.
\end{align}

Under the above large-systems regime, we get the following proposition.
\begin{proposition}\label{Pr:1}
The achievable ergodic rate $I$ in \eqref{eq:ergodic rate replacement} can be asymptotically approximated by
\begin{align}
 \bI = & \log\det\left(\qI_K + \frac{e_2}{\sigma^2} \qR_2 \right) + \log\det\left(\qI_L + e_1\qPsi_2\qR_1 \right) + \log\det \left( \qI_K + e_0 \qPsi_1 \qR_0 \right) \nonumber  \\
  & + \log\det \left(\qI_N + \qOmega \right) - N e_0 \te_0 - N e_1 \te_1 - L e_2 \te_2,\label{eq:Pr1 bI}
\end{align}
where $\{e_0, e_1, e_2, \te_0, \te_1, \te_2\}$ are the unique solutions of the following six equations\footnote{The proof of the existence and uniqueness of the solution is omitted since it is similar to Theorem 2 in \cite{13JSAC-ZhangJun}. The unique solutions of $\{e_k,\te_k\}_{k=0,1,2}$ are calculated by using an iterative algorithm.}
\begin{subequations}\label{eq:e}
\begin{align}
 e_0 = & \frac{1}{N} \tr \left(\qI_N + \qOmega \right)^{-1} \qT_0,   \label{eq:e0}  \\
 e_1 = & \frac{1}{N} \tr \left(\qI_N + \qOmega \right)^{-1} \qT_1 ,  \label{eq:e1}  \\
 e_2 = & \frac{1}{L} \tr \left( e_1\qR_1\qPhi_1^{-1} + e_1^2 \qR_1 \qPi_{11} \qPi_{11}^H\qR_1 + \qPi_{21}\qPi_{21}^H \right) \qT_2 ,  \label{eq:e2}  \\
 \te_0 = & \frac{1}{N} \tr \left(\qPhi_0^{-1} \qPsi_1 - \qPi_{01} \qPi_{01}^H\right)\qR_0, \label{eq:te0}  \\
 \te_1 = & \frac{1}{N} \tr \left( \qPhi_1^{-1}\qPsi_2 - \qPi_{11} \qPi_{11}^H - \qPi_{12}\qPi_{12}^H \right) \qR_1, \label{eq:te1}  \\
 \te_2 = & \frac{1}{L} \tr \left( \qPhi_2^{-1} - \qPi_{31} \qPi_{31}^H - \qPi_{32} \qPi_{32}^H - \qPi_{33} \qPi_{33}^H \right)\qR_2,  \label{eq:te2}
\end{align}
\end{subequations}
with
\begin{subequations}
\begin{align}
  \qPhi_0=& \qI_K + e_0 \qPsi_1 \qR_0,   \\
  \qPhi_1=& \qI_L + e_1 \qPsi_2 \qR_1,   \\
  \qPhi_2=&\sigma^2\qI_K + e_2\qR_2,  \\
  \qPsi_0 = & \qPsi_1 \bqH_0 + \qPhi_2^{-1} \bqH_2 \qPhi_1^{-H} \bqH_1,\\
  \qPsi_1 = &\qPhi_2^{-1} (\qI_K - e_1 \bqH_2 \qR_1 \qPhi_1^{-1} \bqH_2^H \qPhi_2^{-1}), \\
  \qPsi_2 =& \bqH_2^H \qPhi_2^{-1} \bqH_2 + \te_2\qT_2, \\
  \qXi_0 = &\bqH_0 -  e_0 \qR_0 \qPhi_0^{-1} \qPsi_0,\\
  \qOmega = & \bqH_1^H \qPhi_1^{-1}\qPsi_2 \bqH_1 + \qPsi_0^H \qXi_0 + \bqH_0^H \qPhi_2^{-1} \bqH_2 \qPhi_1^{-H} \bqH_1 + \te_0 \qT_0 + \te_1 \qT_1,
\end{align}
\end{subequations}
and
\begin{subequations}
\begin{align}
  \qPi_{01} = &\qPhi_0^{-1} \qPsi_0 (\qI_N + \qOmega )^{-\frac{1}{2}}, \\
 \qPi_{11} = &\qPhi_1^{-1} \bqH_2^H \qPhi_2^{-1} \left( e_0\qR_0\qPhi_0^{-1}\right)^{\frac{1}{2}}, \\
 \qPi_{12} = &\qPhi_1^{-1} \left( \qPsi_2 \bqH_1 + \bqH_2^H \qPhi_2^{-1} \qXi_0 \right) (\qI_N + \qOmega )^{-\frac{1}{2}},  \\
 \qPi_{21} =& \left( \qPhi_1^{-H}\bqH_1 - e_1 \qR_1 \qPhi_1^{-1} \bqH_2^H \qPhi_2^{-1} \qXi_0\right) \left(\qI_N + \qOmega \right)^{-\frac{1}{2}},\\
  \qPi_{31} =& \qPhi_2^{-1} \bqH_2 \left( e_1\qR_1\qPhi_1^{-1}\right)^{\frac{1}{2}},   \\
 \qPi_{32} =& \qPsi_1 \left(e_0\qR_0 \qPhi_0^{-1} \right)^{\frac{1}{2}}, \\
 \qPi_{33} =&  \left( \qPhi_2^{-1} \bqH_2 \qPhi_1^{-H} \bqH_1 + \qPsi_1 \qXi_0 \right) \left(\qI_N + \qOmega \right)^{-\frac{1}{2}}.
\end{align}
\end{subequations}
\end{proposition}
\begin{proof}
 See Appendix A.
\end{proof}

Using Proposition 1, we have the following corollaries for some special case.

\begin{itemize}
  \item Without RIS, i.e., $\qH_1 = \qH_2 = \qzero$, we can obtain the result of single-hop MIMO Rician channel as follows
   \begin{align}
 \bI = \log\det \left( \qI_K + \frac{e_0}{\sigma^2} \qR_0 \right)  + \log\det \left(\qI_N + \te_0 \qT_0 + \bqH_0^H (\sigma^2\qI_K + e_0 \qR_0)^{-1} \bqH_0  \right) - N e_0 \te_0 ,\label{eq:Pr1NoLOS bI}
\end{align}
where
\begin{subequations}\label{eq:eNoLOS}
\begin{align}
 e_0 = & \frac{1}{N} \tr \left(\qI_N + \te_0 \qT_0 + \bqH_0^H (\sigma^2\qI_K + e_0 \qR_0)^{-1} \bqH_0 \right)^{-1} \qT_0,   \label{eq:e0}  \\
 \te_0 = & \frac{1}{N} \tr \left(\sigma^2\qI_K + e_0 \qR_0 + \bqH_0 (\qI_N + \te_0 \qT_0)^{-1} \bqH_0^H \right)^{-1}\qR_0, \label{eq:te0}
\end{align}
\end{subequations}
which agrees with the result in \cite{08TIT-Taricco,10TIT-Dumont}.

  \item Without direct link from BS to UE, i.e., $\qH_0 =\qzero$, we have the following result of two-hop MIMO Rician product channels.
\begin{corollary}\label{Co:1}
The achievable ergodic rate $I$ in \eqref{eq:ergodic rate replacement} can be asymptotically approximated by
\begin{align}
 \bI = & \log\det\left(\qI_K + \frac{e_2}{\sigma^2} \qR_2 \right) + \log\det\left(\qI_L + e_1\qPsi_2\qR_1 \right) \nonumber  \\
  & + \log\det \left(\qI_N + \bqH_1^H \qPhi_1^{-1}\qPsi_2 \bqH_1 + \te_1 \qT_1 \right)  - N e_1 \te_1 - L e_2 \te_2, \label{eq:Co1 bI}
\end{align}
where $\{e_1, e_2, \te_1, \te_2\}$ are the unique solution of the following four equations
\begin{subequations}\label{eq:Co1_e}
\begin{align}
 e_1 = & \frac{1}{N} \tr \left(\qI_N + \bqH_1^H \qPhi_1^{-1}\qPsi_2 \bqH_1 + \te_1 \qT_1 \right)^{-1} \qT_1,  \label{eq:Co1 e1}  \\
 e_2 = & \frac{1}{L} \tr \left( e_1\qR_1 \qPhi_1^{-1} + \qPi \qPi^H \right) \qT_2,  \label{eq:Co1 e2}  \\
 \te_1 = & \frac{1}{N} \tr \left( \qPhi_1^{-1}\qPsi_2 - \qPsi_2 \qPi \qPi^H \qPsi_2 \right) \qR_1, \label{eq:Co1 te1}  \\
 \te_2 = & \frac{1}{L} \tr \left( \qPhi_2^{-1} - \qPhi_2^{-1} \bqH_2 e_1\qR_1\qPhi_1^{-1} \bqH_2^H \qPhi_2^{-1} - \qPhi_2^{-1} \bqH_2 \qPi \qPi^H \bqH_2^H \qPhi_2^{-1} \right)\qR_2,  \label{eq:Co1 te2}
\end{align}
\end{subequations}
and
\begin{align}
  \qPi =& \qPhi_1^{-H}\bqH_1 \left(\qI_N + \bqH_1^H \qPhi_1^{-1}\qPsi_2 \bqH_1 + \te_1 \qT_1 \right)^{-\frac{1}{2}}.
\end{align}
\end{corollary}
Using the relationship between the Shannon transform and the Stieltjes transform \cite{13JSAC-ZhangJun}, i.e., $\frac{\partial \frac{1}{K}\log\det\left( \qI_K+\frac{1}{\sigma^2}\qB_K \right)}{\partial \sigma^2} = \frac{1}{K} \tr \left(\sigma^2\qI_K + \qB_K \right)^{-1}  - \frac{1}{\sigma^2}$ with $\qB_K = \qH_2 \qH_1 \qH_1^H \qH_2^H$, we also obtain a useful result for Stieltjes transform of Rician random matrix product in random matrix theory as follows
\begin{corollary}\label{Co:2}
Defining the Stieltjes transform of $\qB_K$ as $m_{\qB_K}(\omega)\triangleq \frac{1}{K} \tr \left( \qB_K + \omega \qI_K \right)^{-1}$. As $K\rightarrow \infty$, we have
\begin{equation} \label{eq:AidstjCong}
  \Ex\left\{m_{\qB_K}(\omega)\right\} - \frac{1}{K} \tr \left(e_2\qR_2 + \omega \qI_K \right) ^{-1} \xrightarrow{a.s.} 0~~\mbox{for } \omega \in \bbR^{+},
\end{equation}
where $\{e_1, e_2, \te_1, \te_2\}$ are determined by \eqref{eq:Co1_e} and $\qPhi_2 = e_2\qR_2 +\omega\qI_K$.
\end{corollary}

  \item Rayleigh channels, i.e., $\bqH_0 = \bqH_1 = \bqH_2 =\qzero$, we have the following result.
  \begin{corollary}\label{Co:3}
The achievable ergodic rate $I$ in \eqref{eq:ergodic rate replacement} can be asymptotically approximated by
\begin{align}
 \bI = & \log\det\left(\qI_K + \frac{e_1e_2}{\sigma^2} \qR_2 \right) + \log\det \left( \qI_K + e_0 (\sigma^2\qI_K + e_1e_2\qR_2)^{-1} \qR_0 \right) \nonumber  \\
  & + \log\det\left(\qI_L + e_1\te_2\qT_2 \qR_1 \right)+ \log\det \left(\qI_N + \te_0 \qT_0 + \frac{L}{N} e_2\te_2 \qT_1 \right) - N e_0 \te_0 - 2L e_1 e_2\te_2,\label{eq:Co3 bI}
\end{align}
where $\{e_0, e_1, e_2, \te_0, \te_2\}$ are the unique solutions of the following five equations
\begin{subequations}\label{eq:e}
\begin{align}
 e_0 = & \frac{1}{N} \tr \left(\qI_N + \te_0 \qT_0 + \frac{L}{N} e_2\te_2 \qT_1 \right)^{-1} \qT_0,   \label{eq:e0}  \\
 e_1 = & \frac{1}{N} \tr \left(\qI_N + \te_0 \qT_0 + \frac{L}{N} e_2\te_2 \qT_1 \right)^{-1} \qT_1 ,  \label{eq:e1}  \\
 e_2 = & \frac{1}{L} \tr \left(\qI_L + e_1 \te_2 \qT_2 \qR_1\right)^{-1} \qT_2\qR_1 ,  \label{eq:e2}  \\
 \te_0 = & \frac{1}{N} \tr (\sigma^2\qI_K + e_0 \qR_0 + e_1e_2\qR_2)^{-1} \qR_0, \label{eq:te0}  \\
 \te_2 = & \frac{1}{L} \tr (\sigma^2\qI_K + e_0\qR_0 + e_1e_2\qR_2)^{-1}\qR_2.  \label{eq:te2}
\end{align}
\end{subequations}
\end{corollary}

  \item Rayleigh channels without direct link from BS to user, i.e., $\bqH_1 = \bqH_2 = \qH_0 =\qzero$, Corollary 1 and Corollary 3 can be further degraded as the result of MIMO Rayleigh product channels (or MIMO double scattering channels) in \cite{ZhangICC20,11Arxiv-Hoydis}.
\end{itemize}

Combining Proposition 1 and the replacements in \eqref{eq:replacement}, we can get the large-system approximation $\bR(\qQ,\qTheta)$ of $R(\qQ,\qTheta)$ as follows
\begin{align}
 \bR(\qQ,\qTheta) = & \log\det\left(\qI_K + \frac{e_2}{\sigma^2} \qR_2 \right) + \log\det\left(\qI_L + e_1 \qTheta^H (\bqH_2^H \qPhi_2^{-1} \bqH_2 + \te_2\qT_2) \qTheta \qR_1 \right) \nonumber  \\
  &+ \log\det \left( \qI_K + e_0 \qF_2 \qR_0 \right) + \log\det \left(\qI_N + \qF \qQ \right) - N e_0 \te_0 - N e_1 \te_1 - L e_2 \te_2,\label{eq:barR}
\end{align}
where
\begin{subequations} \label{eq:qF123}
\begin{align}
 \qF =& \bqH_1^H \qF_1^{-1} \qTheta^H (\bqH_2^H \qPhi_2^{-1} \bqH_2 + \te_2\qT_2) \qTheta \bqH_1 +\bqH_0^H \qF_2 \bqH_0 +  \bqH_1^H \qF_1^{-1} \qTheta^H \bqH_2^H \qPhi_2^{-1} \bqH_0 \nonumber  \\
     & + \bqH_0^H \qPhi_2^{-1} \bqH_2 \qTheta \qF_1^{-H} \bqH_1  + \te_0 \qT_0 + \te_1 \qT_1  - \left( \bqH_0^H \qF_2 + \bqH_1^H \qF_1^{-1} \qTheta^H \bqH_2^H \qPhi_2^{-1} \right) \nonumber  \\
     & \times e_0 \qR_0 (\qI_K+e_0\qF_2\qR_0)^{-1} \left(\qF_2 \bqH_0 + \qPhi_2^{-1} \bqH_2\qTheta \qF_1^{-H} \bqH_1 \right), \label{eq:qF} \\
 \qPhi_2=&\sigma^2\qI_K + e_2\qR_2,  \\
 \qF_1 =& \qI_L + e_1 \qTheta^H \left(\bqH_2^H \qPhi_2^{-1} \bqH_2 + \te_2\qT_2\right) \qTheta \qR_1,  \\
 \qF_2 = &\qPhi_2^{-1} (\qI_K - e_1 \bqH_2 \qTheta \qR_1 \qF_1^{-1} \qTheta^H \bqH_2^H \qPhi_2^{-1}),
\end{align}
\end{subequations}
in which $\{e_0, e_1, e_2, \te_0, \te_1, \te_2\}$ contain $\qQ$ and $\qTheta$ by using \eqref{eq:e} and \eqref{eq:replacement}.

The above large-system approximation provides very good estimates for the achievable ergodic rate even with finite number of antennas. Therefore, the optimization problem $({\rm P}1)$ in \eqref{eq: ergodic rate maximization P1} can be recast as
\begin{align}\label{eq: ergodic rate maximization P2}
  ({\rm P}2)~~\max_{\qQ,\qTheta} &~~  \bR\left(\qQ,\qTheta\right)  \\
  \mbox{s.t.} &~~ \tr \qQ \leq N P, ~~\qQ \succeq \qzero,  \nonumber \\
  &~~ \qTheta = \diag\{e^{j\theta_1}, e^{j\theta_2}, \ldots, e^{j\theta_L}\},~~ \theta_l\in [0,2\pi). \nonumber
\end{align}
In next few subsections, we propose an alternating method to solve the above optimization problem.

\subsection{Transmit Covariance Matrix Optimization}

For fixed $\qTheta$, we find that $\bR(\qQ,\qTheta)$ in \eqref{eq:barR} is strict concavity with respect to $\qQ$. By using the concave optimal method, the Karush-Kuhn-Tucker (KKT) conditions of the optimization problem in \eqref{eq: ergodic rate maximization P2} are
\begin{equation} \label{eq:P2 KKT}
\left\{
\begin{aligned}
 &- \qF \left(\qI_N + \qF \qQ \right)^{-1} + \qUpsilon - \mu \qI_N = 0, \\
 &\tr (\qUpsilon \qQ) = 0,~\qUpsilon\succeq 0,~\qQ\succeq 0,\\
 & \mu(NP - \tr\qQ) =0,~\mu\geq 0,
\end{aligned} \right.
\end{equation}
where $\qF$ is given by \eqref{eq:qF}, $\mu$ and $\qUpsilon$ are the Lagrange multipliers
associated with the problem constraints.
Thus, the optimization problem in \eqref{eq: ergodic rate maximization P2} is equivalent to the following problem
\begin{align}\label{eq: equivalent maximization P3}
  ({\rm P}3)~~\max_{\qQ} &~~ \log\det\left( \qI_N + \qF \qQ  \right)   \\
  \mbox{s.t.} &~~ \tr \qQ \leq N P, ~~\qQ \succeq \qzero.  \nonumber
\end{align}
We notice that the above problem $({\rm P}3)$ in \eqref{eq: equivalent maximization P3} can be solved by a standard waterfilling procedure and we obtain the following proposition \cite{08TIT-Taricco,12TWC-Wen,13JSAC-ZhangJun,16TIFS-Zhang}.
\begin{proposition}
Let $\qF = \qU_F \qLambda_F \qU_F^H$ is the singular value decomposition of the matrix $\qF$, the asymptotic optimal transmit covariance $\qQ^{\sf opt}$ is given by
\begin{equation}\label{eq: waterfilling Q}
 \qQ^{\sf opt} = \qU_F \qLambda_Q \qU_F^H,
\end{equation}
where $\qLambda_Q$ satisfies $\qLambda_Q = \left( \frac{1}{\mu} \qI_N - \qLambda_F^{-1} \right)^+$ with $(a)^+ = \max\{0, a\}$ and $\mu$ is chosen to satisfy the power constraints $\tr \qQ \leq N P$.
\end{proposition}

\renewcommand{\algorithmicrequire}{\textbf{produce}}
\begin{algorithm}
\setstretch{1.35}
\caption{Optimization Algorithm 1 of $\qQ$}
\label{alg:Framwork}
\begin{algorithmic}[1]
\REQUIRE~~Design of the transmit covariance matrix $\qQ^{\sf opt}$ for fixed $\qTheta$  \\
\STATE {\bfseries Initialize:} $\qQ^{(0)} = \qI_N, e_i^{(0)} = \te_i^{(0)} = 1~(i = 0,1,2)$, and error tolerance $\epsilon$ (we set $\epsilon=10^{-5}$ in the simulations);\\
\STATE {\bfseries for} $t = 0,1,2,\ldots$, do \\

\STATE \quad Given that $\qQ^{(t)}$, calculate $e_i^{(t+1)}$ and $\te_i^{(t+1)} ~(i = 0,1,2)$ based on \eqref{eq:e} and the replacement \eqref{eq:replacement} using $\qQ^{(t)}$ and $\qTheta$; \\
\STATE \quad Calculate $\qF$ based on \eqref{eq:qF} using $\qQ^{(t)}$, $e_i^{(t+1)}$, and $\te_i^{(t+1)} ~(i = 0,1,2)$; \\
\STATE \quad Calculate $\qQ^{(t+1)}$ based on \eqref{eq: waterfilling Q} in Proposition 2; \\
\STATE \quad Calculate $\bR(\qQ^{(t+1)},\qTheta)$ based on \eqref{eq:barR};\\
\STATE \quad {\bfseries Until} $|\bR(\qQ^{(t+1)},\qTheta)-\bR(\qQ^{(t)},\qTheta)| < \epsilon$, obtain $\qQ^{\sf opt} =\qQ^{(t+1)}$; \\
\STATE {\bfseries end for}\\
\end{algorithmic}
\end{algorithm}

It is observed that the asymptotic optimal transmit covariance $\qQ^{\sf opt}$ only depends on the matrix $\qF$ in \eqref{eq:qF}, which contains the statistical CSIT, i.e., $\{\qR_i, \qT_i, \bqH_i, \kappa_i, \Gamma_i, \mbox{~for~} i = 0,1,2\}$. Using Proposition 2, we have the following observations:
\begin{itemize}
  \item Without RIS: When there are no links of BS to RIS and RIS to user (i.e., $\qH_1 = \qH_2 = \qzero$), we have $\qF = \te_0 \qT_0 + \bqH_0^H (\sigma^2\qI + e_0 \qR_0)^{-1} \bqH_0$ which agrees with the result of single-hop Rician MIMO channels in \cite{08TIT-Taricco,10TIT-Dumont}.
  \item Without direct link from BS to user: When $\qH_0 = \qzero$, $\qF = \bqH_1^H \qF_1^{-1} \qTheta^H (\bqH_2^H \qPhi_2^{-1} \bqH_2 + \te_2\qT_2) \qTheta \bqH_1 + \te_1 \qT_1$, where $\qPhi_2=\sigma^2\qI_K + e_2\qR_2$ and $\qF_1 = \qI_L + e_1 \qTheta^H \left(\bqH_2^H \qPhi_2^{-1} \bqH_2 + \te_2\qT_2\right) \qTheta \qR_1$. It is shown that the optimal transmit covariance $\qQ^{\sf opt}$ is affected by all statistical CSIT, i.e., $\{\qR_i, \qT_i, \bqH_i, \kappa_i, \Gamma_i, \mbox{~for~} i = 1,2\}$.
  \item Rayleigh channels and without direct link: When $\bqH_1 = \bqH_2 = \qH_0 =\qzero$, we have $\qF = \te_1 \qT_1$. It means that the optimal transmit covariance $\qQ^{\sf opt}$ only depends on $\qT_1$ and does not depend on $\qT_2$, $\qR_1$, and $\qR_2$. Proposition 2 can be degraded to the result of MIMO double scattering channels for single-user case in \cite{11Arxiv-Hoydis}.
  \item Perfect CSIT: When the Rician factors of three channels $\kappa_t = \infty$ for $t =0,1,2$, the perfect CSIT is available at BS. For this case, we get $\qF = \frac{1}{\sigma^2} (\bqH_0 + \bqH_2 \qTheta  \bqH_1)^H (\bqH_0 + \bqH_2 \qTheta  \bqH_1)$ for fixed $\qTheta$. Thus, Proposition 2 can be degraded to the result of quasi-static block-fading channels in \cite{20JSAC-Zhangshuowen}.
\end{itemize}

Since $\{e_i, \te_i\}_{\forall i}$ are also the functions of $\qQ$, an iterative approach is required to find the optimal solution of $\qQ$ as shown in {\bfseries Algorithm 1}.

\subsection{Diagonal Phase-Shifting Matrix Optimization}

In this subsection, we will focus on the optimization of $\qTheta$. For fixed $\qQ$, the optimization problem \eqref{eq: ergodic rate maximization P2} is also equivalent to the following problem
\begin{align}\label{eq: ergodic rate maximization P4}
  ({\rm P}4)~~\max_{\qTheta} &~~  \bR\left(\qQ,\qTheta\right)  \\
  \mbox{s.t.} &~~ \qTheta = \diag\{e^{j\theta_1}, e^{j\theta_2}, \ldots, e^{j\theta_L}\},~~ \theta_l\in [0,2\pi). \nonumber
\end{align}
The above optimization problem with respect to the phase-shifting matrix is generally non-concave since the norm of diagonal elements in the phase-shifting matrix is 1. A suboptimal solution of $\qTheta$ can be solved using the projected gradient ascent \cite{20TWC-Nadeem}, in which the gradient search is along the monotonically increasing direction of $\bR\left(\qQ,\qTheta\right)$ under the constraint $\qTheta = \diag\{e^{j\theta_1}, e^{j\theta_2}, \ldots, e^{j\theta_L}\}$. Taking the derivative of $\bR\left(\qQ,\qTheta\right)$ with respect to $\vartheta_l = e^{j\theta_l}$, we have $\frac{\partial \bR}{\partial \vartheta_l}$ for $l = 1,2,\ldots,L$, as shown in Appendix B due to the complexity of the expression.

For ease of operation, we first set the step size of gradient ascent as $\Delta$ to proceed this approach. Let $\qvartheta^{(t)}$ = $[\vartheta_1^{(t)}, \vartheta_2^{(t)}, \ldots, \vartheta_L^{(t)}]^T$ and $\qp^{(t)} = \big[\frac{\partial \bR}{\partial \vartheta_1}^{(t)}, \frac{\partial \bR}{\partial \vartheta_2}^{(t)}, \ldots, \frac{\partial \bR}{\partial \vartheta_L}^{(t)} \big]^T$ denote the phase result and the computed ascent direction at the $t$-th step, respectively. Then we have a new phase-shifting vector as
\begin{align}
  \qvartheta^{(t+1)} = \exp \left(j\arg \left(\qvartheta^{(t)} + \Delta \qp^{(t)}\right)\right).\label{eq:update qvartheta}
\end{align}
Notice that the above new phase-shifting vector satisfies the constraint $|\vartheta_l^{(t+1)}|=1$ for $l = 1,2,\ldots,L$. We also obtain the corresponding phase-shirting matrix $\qTheta^{(t+1)} = \diag(\qvartheta^{(t+1)})$. The proposed iterative approach can be described as {\bfseries Algorithm 2}.

\renewcommand{\algorithmicrequire}{\textbf{produce}}
\begin{algorithm}
\setstretch{1.35}
\caption{Optimization Algorithm 2 for $\qTheta$}
\label{alg:Framwork}
\begin{algorithmic}[1]
\REQUIRE~~Design of the diagonal phase-shifting matrix $\qTheta^{\sf opt}$ for fixed $\qQ$ \\
\STATE {\bfseries Initialize:} $\qTheta^{(0)}$ is randomly generated by setting that the phases $\{\theta_l\}_{\forall l}$ are uniform and independent distribution in $[0,2\pi)$, $e_i^{(0)} = \te_i^{(0)} = 1~(i = 0,1,2)$, and error tolerance $\epsilon$;\\
\STATE {\bfseries for} $t = 0,1,2,\ldots$, do \\
\STATE \quad Given that $\qTheta^{(t)}$ is available. Calculate $e_i^{(t+1)}$ and $\te_i^{(t+1)} ~(i = 0,1,2)$ based on \eqref{eq:e} and the replacement \eqref{eq:replacement} using $\qQ$ and $\qTheta^{(t)}$; \\
\STATE \quad Calculate $\qp^{(t+1)}$ based on \eqref{eq:derivative barR}; \\
\STATE \quad Calculate $\qvartheta^{(t+1)}$ based on \eqref{eq:update qvartheta} and $\qTheta^{(t+1)} = \diag(\qvartheta^{(t+1)})$;\\
\STATE \quad Calculate $\bR(\qQ,\qTheta^{(t+1)})$ based on \eqref{eq:barR};\\
\STATE \quad {\bfseries Until} $|\bR(\qQ,\qTheta^{(t+1)})-\bR(\qQ,\qTheta^{(t)})| < \epsilon$, obtain $\qTheta^{\sf opt} = \qTheta^{(t+1)}$; \\
\STATE \quad {\bfseries end for} \\
\end{algorithmic}
\end{algorithm}

\subsection{Proposed Algorithm}

\renewcommand{\algorithmicrequire}{\textbf{produce}}
\begin{algorithm}
\setstretch{1.35}
\caption{Alternating Optimization Algorithm 3 for Problem (P2)}
\label{alg:Framwork}
\begin{algorithmic}[1]
\REQUIRE~~Design of the transmit covariance matrix $\qQ^{\sf opt}$ and diagonal phase-shifting matrix $\qTheta^{\sf opt}$  \\
\STATE {\bfseries Initialize:} $\qQ^{(0)} = \qI_N, \qTheta^{(0)}$ is randomly generated by setting that the phases $\{\theta_l\}_{\forall l}$ are uniform and independent distribution in $[0,2\pi)$, and error tolerance $\epsilon$;\\
\STATE {\bfseries for} $n = 0,1,2,\ldots$, do \\
\STATE \quad Given that $\qQ^{(n)}$ is available. Calculate $\qTheta^{(n+1)}$ based on {\bfseries Algorithm 1};\\
\STATE \quad Calculate $\qQ^{(n+1)}$ based on {\bfseries Algorithm 2}; \\
\STATE \quad Calculate $\bR(\qQ^{(n+1)},\qTheta^{(n+1)})$ based on \eqref{eq:barR};\\
\STATE \quad {\bfseries Until} $|\bR(\qQ^{(n+1)},\qTheta^{(n+1)})-\bR(\qQ^{(n)},\qTheta^{(n)})| < \epsilon$, obtain $\qQ^{\sf opt} =\qQ^{(n+1)}$ and $\qTheta^{\sf opt} = \qTheta^{(n+1)}$; \\
\STATE {\bfseries end for} \\
\end{algorithmic}
\end{algorithm}

In the above two subsections, we have done the optimization of transmit covariance matrix and diagonal phase-shifting matrix by using the alternating method, respectively. Now, we present the complete alternating optimization algorithm to find the above two matrices in {\bfseries Algorithm 3}. We first initialize the transmit covariance matrix $\qQ^{(0)} = \qI_N$ and then obtain the diagonal phase-shifting matrix $\qTheta^{(1)}$ according to \eqref{eq:update qvartheta}. Next, for fixed $\qTheta^{(1)}$, we update the optimal $\qQ^{(1)}$ according to \eqref{eq: waterfilling Q}. By iteratively calculating $\qTheta^{(n+1)}$ and $\qQ^{(n+1)}$, our algorithm will complete until convergence is satisfied.

Next, we discuss the convergence of our proposed algorithm. Firstly, the six parameters $e_i$ and $\te_i~(i = 0,1,2)$ are determined by functions \eqref{eq:e}, but those functions are implicit functions. The existence and uniqueness of the solution to \eqref{eq:e} should be considered in here, however, we omit the proof of the existence and uniqueness since the proof is similar to the previous works. By using the subsequence approach, the existence of the solution can be obtained in \cite{11TIT-Couillet}. By reduction to absurdity, the uniqueness of the solution can be proved in \cite{13JSAC-ZhangJun}. Secondly, $\qQ$ in \eqref{eq: waterfilling Q} is an optimal solution since $\bR(\qQ,\qTheta)$ in \eqref{eq:barR} is strict concavity with respect to $\qQ$ for fixed $\qTheta$. Finally, we note that $\qTheta$ in \eqref{eq:update qvartheta} is not a global optimal solution and only is locally optimal due to the non-convex constraint $\qTheta = \diag\{e^{j\theta_1}, e^{j\theta_2}, \ldots, e^{j\theta_L}\}$. However, the calculation of $\qTheta$ will converge since the gradient search is along the monotonically increasing direction of $\bR\left(\qQ,\qTheta\right)$ at each step. Hence, the proposed algorithm is guaranteed to converge.

\section{Simulation results}

Numerical simulations are conducted in this section to compare the analytical result $\bR(\qQ,\qTheta)$ in \eqref{eq:barR} with the Monte Carlo simulation result of the achievable ergodic rate $R(\qQ,\qTheta)$ in \eqref{eq:ergodic rate} and examine the effectiveness of the proposed algorithm for the optimization of transmit covariance matrix $\qQ$ and diagonal phase-shifting matrix $\qTheta$.
We consider a simulation setup in Fig \ref{fig:2}, where the coordinates of BS, RIS, and user are $(0; 10\mbox{m})$, $(0; d(\mbox{m}))$, and $(80\mbox{m};10\mbox{m})$, respectively, the transmit power at BS is $P = 10$dBm, the bandwith is $B=10$MHz, the noise power is $-94$dBm, the path loss model $\Gamma_i (d_{TR})[dB] = G_t +G_r -37.5-22\log_10(d_{TR}/1m)$ (where $G_t=G_r=5dBi$ denote the antenna gains (in dBi) at the transmitter and receiver, respectively) \cite{20WCL-Bjornson,20TWC-Nadeem}, and $\kappa_i=1$ for $i = 0,1,2$.
Without loss of generality, the LoS components $\bqH_i$ for $i=0,1,2$ are set to be all-one matrices, i.e., $[\bqH]_{m,n} = 1$ for $\forall m,n$ \cite{12TWC-Wen}, and the transmit and receive correlation matrices (i.e., $\qT_i$ and $\qR_i$ for $i=0,1,2$) are generated by \cite{03TIT-Moustakas}
\begin{equation} \label{eq: arrayPattern}
[\qT~\mbox{or}~\qR ]_{m,n} =\int_{-180}^{180}{ \frac{d\phi}{\sqrt{2\pi\delta^2}}e^{2\pi{\sf j}d_s(m-n)\sin\left(\frac{\pi\phi}{180}\right)-\frac{(\phi- \eta)^2}{2\delta^2}} },
\end{equation}
where $m, n$ are the indexes of antennas, $d_s$ is the relative antenna spacing (in wavelengths), $\varphi_n \in [-\pi/2, \pi/2)$ is the physical angle (in the plane of the arrays), $\eta$ is the mean angle, and $\delta$ is the root-mean-square angle spread. The relative antenna spacing $d_s =1$, the mean angle and the root-mean-square angle spread are set as $(\eta_0^{R}, \delta_0^{R}) = (0^\circ,30^\circ),  (\eta_0^{T}, \delta_0^{T}) = (10^\circ,5^\circ), (\eta_1^{R}, \delta_1^{R}) = (0^\circ,20^\circ), (\eta_1^{T}, \delta_1^{T}) = (0^\circ,5^\circ), (\eta_2^{R}, \delta_2^{R}) = (0^\circ,5^\circ), (\eta_0^{T}, \delta_0^{T}) = (0^\circ,30^\circ)$.

\begin{figure*}
\centering
\begin{psfrags}%
\psfragscanon%
\small{ \resizebox{10cm}{!}{\includegraphics{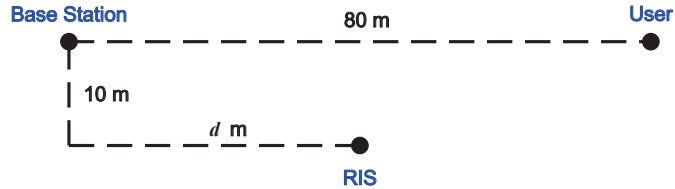}}}
\end{psfrags}%
\caption{The simulation setup where $d$ is a variable}\label{fig:2}
\end{figure*}

\begin{figure*}
\centering
\begin{psfrags}%
\psfragscanon%
\small{ \resizebox{10cm}{!}{\includegraphics{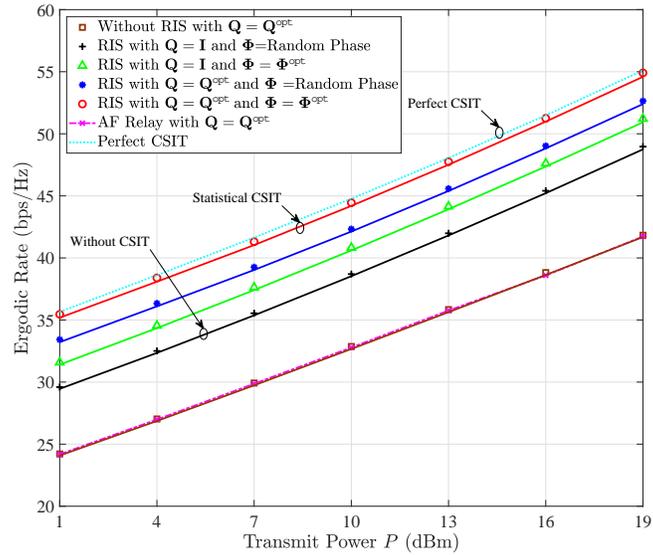}}}
\end{psfrags}%
\caption{Achievable ergodic rates and their approximations \emph{vs}. the transmit power $P$ at BS with $N = L = K = 8$ and $d = 40$m for various schemes. The markers and solid curves respectively denote the simulation and analytic results.}\label{fig:3}
\end{figure*}

\begin{figure*}
\centering
\begin{psfrags}%
\psfragscanon%
\small{ \resizebox{10cm}{!}{\includegraphics{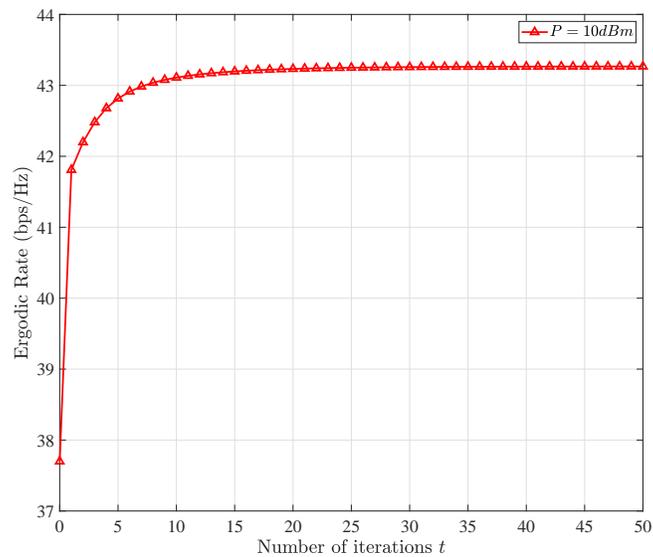}}}
\end{psfrags}%
\caption{Convergence of the proposed algorithm with $P = 10$dBm, $N = L = K = 8$, and $d=40$m.}\label{fig:4}
\end{figure*}

Fig.~\ref{fig:3} shows the results of the achievable ergodic rate and its large-system approximation versus the transmit power $P$ at BS of various schemes for the RIS-assisted MIMO system with $N = L = K = 8$ and $d = 40$m. The analytical results (solid curves) perfectly agree with the simulation results (markers) achieved by Monte Carlo averaging even for small number of antennas. Moreover, the RIS-assisted MIMO system outperforms the MIMO system without RIS and the proposed scheme (red curve) achieves superior performance to other schemes.
We also compare the performance of three different schemes: Perfect, statistical, and without CSIT, where the scheme of perfect CSIT means that the design of transmit covariance matrix $\qQ$ and diagonal phase-shifting matrix $\qTheta$ is based on the instantaneous CSIT $\qH_i (i=0,1,2)$ by $10000$ Monte-Carlo averaging, the scheme of statistical CSIT proposed in this paper is based on statistical CSIT $\{\qR_i, \qT_i, \bqH_i, \kappa_i, \Gamma_i, \mbox{~for~} i = 0,1,2\}$, and when without CSIT, $\qQ =\qI$ and $\qTheta$ is randomly generated by setting that the phases $\{\theta_l\}_{\forall l}$ are uniform and independent distribution in $[0,2\pi)$.
Obviously, the scheme without CSIT is the easiest one to implement and the complexity of scheme with perfect CSIT is much higher than the scheme with statistical CSIT since perfect CSIT is difficult to obtain at BS. However,
It can be observed that the gap between the schemes of perfect and statistical CSIT is very small and the proposed scheme is obviously superior to the scheme without CSIT.
Due to the low complexity of the scheme without CSIT, some recent works in \cite{21WCL-Nadeem,20PIMRC-Psomas} considered other cases with random phase shifts without CSIT, which can also perform better.

In Fig.~\ref{fig:3}, the performance of RIS-assisted MIMO system is also compared to the benchmark schemes of amplify-and-forword (AF) relay system \cite{12TWC-Wen} equipped with the same number of antennas $L = 8$ as RIS. The transmit power at BS $P_s$ and the transmit power at relay $P_r$ satisfy $P_s + P_r =P$ and the optimal power allocation is found through numerical exhaustive search \cite{20TWC-Nadeem}. We observe that the performance of AF relay system almost identically to the performance without RIS. It is because that all the power is allocated to BS since the gain of relay link is smaller than the gain of direct link from BS to user.

In Fig.~\ref{fig:4}, we show the large-system approximation of the achievable ergodic rate versus the number of iterations for the RIS-assisted MIMO system with $P = 10dBm$, $N = L = K = 8$, and $d = 40$m. It can be observed that the proposed algorithm converges with in $25$ iterations, which confirms the convergence of the proposed algorithm.

\begin{figure*}
\centering
\begin{psfrags}%
\psfragscanon%
\small{ \resizebox{10cm}{!}{\includegraphics{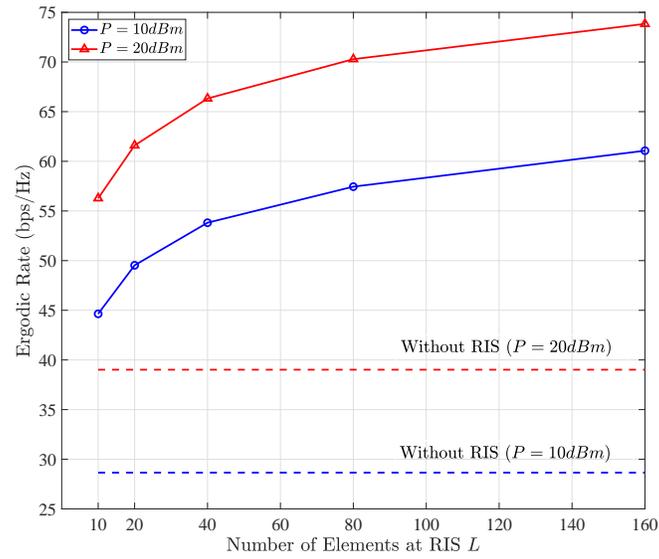}}}
\end{psfrags}%
\caption{Achievable ergodic rates and their approximations \emph{vs}. $L$ with $N = K = 8$ and $d = 40$m for various transmit power.}\label{fig:5}
\end{figure*}

Fig.~\ref{fig:5} depicts the achievable ergodic rate versus the number of elements at RIS $(L)$ with $N = K = 10$ and $d = 40$m. We observe that the achievable ergodic rate increases with increasing the number of elements at RIS and the RIS-assisted MIMO system outperforms the MIMO system without RIS.

\begin{figure*}
\centering
\begin{psfrags}%
\psfragscanon%
\small{ \resizebox{10cm}{!}{\includegraphics{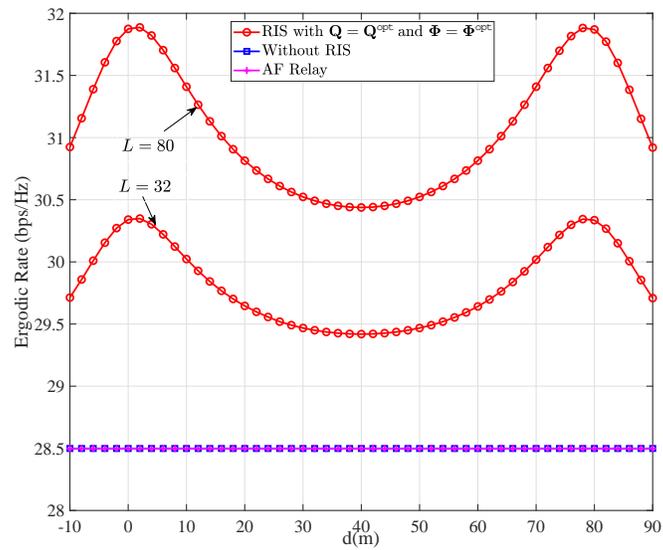}}}
\end{psfrags}%
\caption{Achievable ergodic rates \emph{vs}. $d$ with $N = K = 8$ for various schemes.}\label{fig:6}
\end{figure*}

Finally, Fig.~\ref{fig:6} compares the achievable ergodic rate versus the distance $d$ of three schemes with $N  = K = 8$ and $L = 32$ or $80$. All curves are employed the optimal schemes, i.e., $\qQ = \qQ^{\sf opt}$ and $\qTheta = \qTheta^{\sf opt}$. The result shows that when RIS is close to BS or user, the achievable ergodic rate can achieve its maximum; when RIS is located in the middle between BS and user, the achievable ergodic rate is minimum. This provides a reference for the deployment of RIS. We also see that the performance of AF relay system almost identically to the performance without RIS because all the power is allocated to BS.

\section{Conclusion}

Downlink RIS-assisted MIMO wireless communication system, which comprises three communication links of Rician channel, including BS to RIS, RIS to user, and BS to user, is investigated by exploiting the statistical CSIT. The large-system approximation of the achievable ergodic rate using the replica method is derived. The large-system approximation result is used to design the optimal transmit covariance and diagonal phase-shifting matrices at BS and RIS, respectively, to maximize the achievable ergodic rate. Numerical results reveal that the large-system approximation provides reliable performance predictions even for small antenna systems and verify the effectiveness of the proposed algorithm. Many developments are still ongoing due to the wide application of RIS to various wireless communication scenarios.

\section*{Appendix}
\subsection{Proof of Proposition \ref{Pr:1}}\label{Appendix: Proof of Proposition 1}

\subsubsection{Mutual Information}

Since $\qz\in \bbC^K$ is the noise vector whose entries consist of independent zero-mean circularly symmetric complex Gaussian with variance $\sigma^2$, the receive signal can be described as the following conditional pdf
\begin{align} \label{eq:conditional pdf}
 p(\qy|\qs,\calqH) = \frac{1}{(\pi\sigma^2)^K} e^{-\frac{1}{\sigma^2}\|\qy- (\qH_0+\qH_2\qH_1)\qs\|^2},
\end{align}
where $\calqH = \{\qH_i,i=0,1,2\}$.

Using the Bayes formula, the posteriori distribution can be expressed as
\begin{align}
 p(\qs|\qy,\calqH) = \frac{p(\qy|\qs,\calqH)p(\qs)}{p(\qy|\calqH)},
\end{align}
where $p(\qs)$ is an input distribution. For fixed $p(\qs)$ and $\calqH$, we have the following conditional mutual information of MIMO RIS-based channel
\begin{align}\label{eq:conditional mutual information}
 I(\qs;\qy|\calqH) = \Ex_{\{\qs,\qy\}} \left\{ \log \frac{p(\qy|\qs,\calqH)}{p(\qy|\calqH)}\bigg|\calqH \right\}.
\end{align}
Thus, the mutual information can be expressed by
\begin{align}
 I(\qs;\qy) & = \Ex_{\{\calqH\}} \left\{I(\qs;\qy|\calqH)\right\} \nonumber \\
            & = - \Ex_{\{\qy,\calqH\}} \left\{ \log \Ex_\qs \left\{ e^{-\frac{1}{\sigma^2}\|\qy- (\qH_0+\qH_2\qH_1)\qs\|^2} \right\} \right\} - K.  \label{eq:conditional mutual information}
\end{align}

For brevity of notation, we define
\begin{align}\label{eq:F}
 F \triangleq - \Ex_{\{\qy,\calqH\}} \left\{ \log T(\qy,\calqH) \right\},
\end{align}
where
\begin{align}\label{eq:T}
 T(\qy,\calqH) \triangleq \Ex_\qs \left\{e^{-\frac{1}{\sigma^2}\|\qy-(\qH_0+\qH_2\qH_1) \qs\|^2} \right\}.
\end{align}

By the mathematical derivation\footnote{For positive random variable $X$, we have $\lim_{r\rightarrow 0} \frac{\partial}{\partial r} \log \Ex \{ X^r\} = \lim_{r\rightarrow 0} \frac{\Ex \{X^r \log X\}}{\Ex \{ X^r\}} = \Ex \{\log X\}$}, \eqref{eq:F} can be rewritten as
\begin{align}\label{eq:F2}
 F = - \lim_{r\rightarrow 0} \frac{\partial}{\partial r} \log \Ex_{\{\qy,\calqH\}} \left\{ T^r(\qy,\calqH) \right\}.
\end{align}
From \eqref{eq:T}, we obtain
\begin{align}\label{eq:ETr}
 \Ex_{\{\qy,\calqH\}} \left\{ T^r(\qy,\calqH) \right\} =  \Ex_{\{\qS,\calqH\}} \left\{ \int d \qy \frac{1}{(\pi\sigma^2)^K} \prod_{\alpha=0}^r e^{-\frac{1}{\sigma^2}\|\qy- (\qH_0+\qH_2\qH_1) \qs^{(\alpha)}\|^2} \right\},
\end{align}
where $\qs^{(\alpha)}$ denotes the $\alpha$-th replica signal vector $\qs$ with the input distribution $p(\qs)$ for $\alpha = 0,1,2,\dots,r$. Let $\qS = [\qs^{(0)},\qs^{(1)},\ldots,\qs^{(r)}]$.

\subsubsection{Taking Expectation over $\qH_2$}

Define
\begin{align}
 \qv_2^{(\alpha)} = (\qH_0+\qH_2\qH_1) \qs^{(\alpha)}
\end{align}
and $\qV_2 = [\qv_2^{(0)},\qv_2^{(1)},\ldots,\qv_2^{(r)}]$. From the central limit theorem, as $N\rightarrow \infty$, $\qV_2$ converges to Gaussian random matrix with mean $\bqV_2 = [(\qH_0 +\bqH_2 \qH_1) \qs^{(0)},(\qH_0 +\bqH_2 \qH_1) \qs^{(1)},\ldots,(\qH_0 +\bqH_2 \qH_1) \qs^{(r)}]$ and covariance $\qC_2 \otimes \qR_2$, where $\qC_2 \in \bbC^{(r+1)\times (r+1)}$ is matrix with entries given by 
\begin{align}
 [\qC_2]_{\alpha,\beta} &= \frac{1}{L} \left(\qH_1 \qs^{(\alpha)}\right)^H \qT_2 \left(\qH_1 \qs^{(\beta)}\right),
\end{align}
for $\alpha,\beta = 0,\ldots,r$.

Thus, \eqref{eq:ETr} can be rewritten as
\begin{align}
 \Ex_{\{\qy,\calqH\}} \left\{ T^r(\qy,\calqH) \right\} =& \int \Ex_{\calqH} \left\{ p(\qy|\calqH) T^r(\qy,\calqH) \right\} d\qy  \nonumber\\
=&\Ex_{\qS}  \left\{ \int e^{\calG_2^{(r)}(\qC_2)} d \mu_2^{(r)}(\qC_2) \right\} + \calO(1), \label{eq:ETr2}
\end{align}
where
\begin{align}\label{eq:G_C}
 &\calG_2^{(r)}(\qC_2) = \log \Ex_{\qV_2} \left\{ \int d\qy \frac{1}{(\pi\sigma^2)^K} \prod_{\alpha=0}^r e^{-\frac{1}{\sigma^2} \|\qy-\qv_2^{(\alpha)}\|^2} \right\},
\end{align}
and
\begin{align}\label{eq:mu_C}
 \mu_2^{(r)}(\qC_2) =& \prod_{\alpha,\beta=0}^r \delta\left( \left(\qH_1 \qs^{(\alpha)}\right)^H \qT_2 \left(\qH_1 \qs^{(\beta)}\right) - L [\qC_2]_{\alpha,\beta} \right) 
\end{align}
is the probability measure of $\qC_2$, $\delta\left(\cdot\right)$ denotes the Dirac function, and $\calO(\cdot)$ being a constant as $N\rightarrow \infty$.

Now, we focus on \eqref{eq:G_C} and \eqref{eq:mu_C}. First, after integrating over $\qy$ and taking expectation over $\qV_2$ by Lemma \ref{Lemma:GIHST} and the matrix formula ${\tr(\qA\qB\qC\qD) = (\Ve(\qA^H))^H (\qD^H\otimes\qB) \Ve(\qC)}$, \eqref{eq:G_C} can be rewritten as
\begin{align}
 \calG_2^{(r)}(\qC_2)=& \log \Ex_{\qV_2} \left\{ \int d\qy \frac{1}{(\pi\sigma^2)^K} e^{-\frac{1}{\sigma^2}\sum_{\alpha=0}^r \|\qy-\qv_2^{(\alpha)}\|^2} \right\} \nonumber \\
  =& \Ex_{\qH_0,\qH_1} \left\{ \tr\left( \qA_2 \qS^H (\qH_0^H+\qH_1^H \bqH_2^H) \qB_2 (\qH_0+\bqH_2 \qH_1) \qS \right) \right\} \nonumber \\
 &- \log\det\left(\qI + \qC_2 \qSigma \otimes \qR_2 \right) - K \log\left(1+r\right), \label{eq:G_C_2}
\end{align}
where 
\begin{align}\label{eq:Sigma}
 \qSigma \triangleq \frac{1}{\sigma^2(r+1)}\left((r+1)\qI_{r+1} - \qone_{r+1} \qone_{r+1}^T\right),
\end{align}
\begin{align}\label{eq:AgBg}
  \qA_2\otimes \qB_2 
  &= - (\qSigma \otimes\qI) (\qC_2 \qSigma \otimes \qR_2 + \qI)^{-1}.
\end{align}

Next, we can rewrite \eqref{eq:mu_C} as
\begin{align}\label{eq:muC_2}
 \mu_2^{(r)}(\qC_2) = e^{-\calR_2^{(r)}(\qC_2)+ \calO(1)},
\end{align}
where $\calR_2^{(r)}(\qC_2)$ is the rate measure of $\mu_2^{(r)}(\qC_2)$ and is given by
\begin{align}\label{eq:RC1}
 &\calR_2^{(r)}(\qC_2) =  \max_{\tqC_2} \left\{ L \tr(\tqC_2\qC_2) - \log \Ex_{\{\qS,\qH_1\}} \left\{e^{ \tr \left(\tqC_2 \qS^H \qH_1^H \qT_2 \qH_1 \qS \right)} \right\} \right\},
\end{align}
where $\tqC_2 \in \bbC^{(r+1)\times(r+1)}$ being a symmetric matrix.

Substituting \eqref{eq:G_C_2} and \eqref{eq:muC_2} into \eqref{eq:ETr2} yields $\int  e^{\calG_2^{(r)}(\qC_2) - \calR_2^{(r)}(\qC_2) + \calO(1)} d \qC_2$. Therefore, as $N\rightarrow\infty$, the integration over $\qC_2$ can be performed via the saddle point method, yields
\begin{align}\label{eq:ETrmaxC1}
 -\lim_{N\rightarrow\infty} \log \Ex_{\{\qy,\calqH\}} \left\{ T^r(\qy,\calqH) \right\} = \min_{\qC_2} \left\{ - \calG_2^{(r)}(\qC_2) + \calR_2^{(r)}(\qC_2) \right\}.
\end{align}

\subsubsection{Taking Expectation over $\qH_1$}
From \eqref{eq:G_C_2} and \eqref{eq:RC1}, we find that $\calG_2^{(r)}(\qC_2)$ and $\calR_2^{(r)}(\qC_2)$ still contain the random component of channel. We further define
\begin{align}
 \qv_0^{(\alpha)} &= \qH_0 \qs^{(\alpha)} \mbox{~~and~~}
 \qv_1^{(\alpha)} = \qH_1 \qs^{(\alpha)}. \nonumber
\end{align}
Leting $\qV_0 = [\qv_0^{(0)},\qv_0^{(1)},\ldots,\qv_0^{(r)}]$ and $\qV_1 = [\qv_1^{(0)},\qv_1^{(1)},\ldots,\qv_1^{(r)}]$. As $N\rightarrow \infty$, $\qV_0$ converges to Gaussian random matrix with mean $\bqV_0 = [\bqH_0 \qs^{(0)}, \bqH_0 \qs^{(1)},\ldots, \bqH_0 \qs^{(r)}]$ and covariance $\qC_0 \otimes \qR_0$, $\qV_1$ converges to Gaussian random matrix with mean $\bqV_1 = [\bqH_1 \qs^{(0)}, \bqH_1 \qs^{(1)},\ldots, \bqH_1 \qs^{(r)}]$ and covariance $\qC_1 \otimes \qR_1$, where $\qC_0 \in \bbC^{(r+1)\times (r+1)}$ and $\qC_1 \in \bbC^{(r+1)\times (r+1)}$ are matrices with entries given, respectively, by
\begin{align}
 [\qC_0]_{\alpha,\beta} &= \frac{1}{N} \left(\qs^{(\alpha)}\right)^H \qT_0 \qs^{(\beta)}, \\
 [\qC_1]_{\alpha,\beta} &= \frac{1}{N} \left(\qs^{(\alpha)}\right)^H \qT_1 \qs^{(\beta)},
\end{align}
for $\alpha,\beta = 0,1,\ldots,r$.

Similar to \eqref{eq:ETr2}, by taking the expectation with respect to $\qH_1$, the exponent terms of \eqref{eq:G_C_2} and \eqref{eq:RC1} can be rewritten as
\begin{align}
 &\Ex_{\{\qS,\qH_0,\qH_1\}} \left\{e^{\tr\left( \qA_2 \qS^H (\qH_0^H+\qH_1^H \bqH_2^H) \qB_2 (\qH_0+\bqH_2 \qH_1) \qS \right)+ \tr \left(\tqC_2 \qS^H \qH_1^H \qT_2 \qH_1 \qS \right)} \right\} \nonumber  \\
 =&\Ex_{\{\qS,\qH_0\}}  \left\{ \int e^{\calG_1^{(r)}(\qC_1)} d \mu_1^{(r)}(\qC_1) \right\} + \calO(1), \label{eq:ETr3}
\end{align}
where
\begin{align}
 \calG_1^{(r)}(\qC_1) = & \tr \left( \qA_1 \right) - \log\det \left(\qI - \left(\qA_2 \otimes \bqH_2^H \qB_2 \bqH_2 +  \tqC_2 \otimes \qT_2 \right)\left( \qC_1 \otimes \qR_1 \right)\right), \label{eq:ETr3} \\
 \mu_1^{(r)}(\qC_1) = & \prod_{\alpha,\beta=0}^r \delta\left( {\qs^{(\alpha)}}^H \qT_1 \qs^{(\beta)} - N [\qC_1]_{\alpha,\beta} \right) \label{eq:mu2}
\end{align}
with
\begin{align}\label{eq:A0B0}
\qA_1 = & \qH_0^H\qB_2\qH_0\qS \qA_2 \qS^H + \qH_0^H \qB_2 \bqH_2 \qR_1 \bqH_2^H \qB_2 \qD_1^{-1} \qH_0 \qS \qA_2 \qC_1 \qS \qA_2 \qS^H \nonumber  \\
 &+ \qH_0^H \qB_2 \bqH_2\bqH_1 \qS \qA_2 \qS^H  + \bqH_1^H \bqH_2^H \qB_2 \qH_0 \qS \qA_2 \qS^H  \nonumber  \\
 &+ \qH_0^H \qB_2 \bqH_2 \qR_1 \bqH_2^H \qB_2 \bqH_2 \qD_1^{-1} \bqH_1 \qS \qA_2 \qC_1 \qA_2 \qS^H
  + \qH_0^H \qB_2 \bqH_2 \qR_1 \qT_2 \qD_1^{-1} \bqH_1 \qS \qA_2 \qC_1 \tqC_2 \qS^H   \nonumber  \\
 &+ \bqH_1^H \bqH_2^H \qB_2 \bqH_2 \qR_1\bqH_2^H \qB_2  \qD_1^{-1} \qH_0 \qS \qA_2 \qC_1 \qA_2 \qS^H
  + \bqH_1^H \qT_2 \qR_1\bqH_2^H \qB_2 \qD_1^{-1} \qH_0 \qS \qA_2 \qC_1 \tqC_2 \qS^H \nonumber  \\
 &+ \bqH_1^H \bqH_2^H \qB_2 \bqH_2 \bqH_1 \qS \qA_2 \qS^H  +  \bqH_1^H \qT_2 \bqH_1 \qS \tqC_2 \qS^H  \nonumber  \\
 &+ \bqH_1^H \bqH_2^H \qB_2 \bqH_2 \qR_1 \bqH_2^H \qB_2 \bqH_2    \qD_1^{-1} \bqH_1 \qS \qA_2 \qC_1 \qA_2 \qS^H  \nonumber  \\
 &+ \bqH_1^H \bqH_2^H \qB_2 \bqH_2 \qR_1 \qT_2  \qD_1^{-1} \bqH_1 \qS \qA_2 \qC_1 \tqC_2 \qS^H  \nonumber  \\
 &+ \bqH_1^H \qT_2 \qR_1 \bqH_2^H \qB_2 \bqH_2  \qD_1^{-1} \bqH_1  \qS \qA_2 \qC_1 \tqC_2 \qS^H
  + \bqH_1^H \qT_2 \qR_1 \qT_2 \qD_1^{-1} \bqH_1  \qS \tqC_2 \qC_1 \tqC_2 \qS^H,\\
 \qD_1=& \qI-(\qA_2 \otimes \bqH_2^H \qB_2 \bqH_2 + \tqC_2 \otimes \qT_2)(\qC_1 \otimes \qR_1).
\end{align}

\subsubsection{Taking Expectation over $\qH_0$}
Then, taking the expectation with respect to $\qH_0$, we further obtian
\begin{align}
 &\Ex_{\{\qS,\qH_0,\qH_1\}} \left\{e^{\tr\left( \qA_2 \qS^H (\qH_0^H+\qH_1^H \bqH_2^H) \qB_2 (\qH_0+\bqH_2 \qH_1) \qS \right)+ \tr \left(\tqC_2 \qS^H \qH_1^H \qT_2 \qH_1 \qS \right)} \right\} \nonumber  \\
 =&\Ex_{\qS}  \left\{ \int e^{\calG_0^{(r)}(\qC_0,\qC_1)} d \mu_0^{(r)}(\qC_0,\qC_1) \right\} + \calO(1), \label{eq:ETr44}
\end{align}
where
\begin{align}
 \calG_0^{(r)}(\qC_0,\qC_1) = & \tr \left( \qA_0 \qS^H \qB_0 \qS \right) - \log\det \left(\qI - \left(\qA_2 \otimes \bqH_2^H \qB_2 \bqH_2 +  \tqC_2 \otimes \qT_2 \right)\left( \qC_1 \otimes \qR_1 \right)\right)  \nonumber  \\
 & - \log\det \left(\qI - \qD_2( \qC_0 \otimes \qR_0)\right)\label{eq:ETr3}\\
 \mu_0^{(r)}(\qC_0,\qC_1) = &\prod_{\alpha,\beta=0}^r \delta\left( {\qs^{(\alpha)}}^H \qT_0 \qs^{(\beta)} - N [\qC_0]_{\alpha,\beta} \right)\delta\left( {\qs^{(\alpha)}}^H \qT_1 \qs^{(\beta)} - N [\qC_1]_{\alpha,\beta} \right)\label{eq:mu2}
\end{align}
with
\begin{align}\label{eq:A0B0}
\qA_0 \otimes \qB_0=& (\qI \otimes \bqH_1^H) (\qA_2 \otimes \bqH_2^H \qB_2 \bqH_2 + \tqC_2 \otimes \qT_2) \qD_1^{-1} (\qI \otimes \bqH_1)  \nonumber  \\
 &+ [ (\qI\otimes\bqH_0^H) \qD_2 +(\qI \otimes \bqH_1^H) \qD_1^{-1} (\qA_2\otimes \bqH_2^H\qB_2)] (\qC_0\otimes\qR_0)\qD_0^{-1}\nonumber  \\
 &~~\times [ (\qA_2\otimes \qB_2 \bqH_2) \qD_1^{-1} (\qI \otimes \bqH_1) +  \qD_2 (\qI\otimes\bqH_0)] \nonumber  \\
 &+  (\qI\otimes\bqH_0^H) \qD_2 (\qI\otimes\bqH_0) +(\qI\otimes\bqH_1^H) \qD_1^{-1} (\qA_2\otimes\bqH_2^H \qB_2) (\qI\otimes\bqH_0) \nonumber  \\
 &+ (\qI\otimes\bqH_0^H) (\qA_2\otimes\qB_2 \bqH_2)\qD_1^{-1}(\qI\otimes\bqH_1),  \\
 \qD_0=& \qI - \qD_2 ( \qC_0 \otimes \qR_0), \\
 \qD_2=& (\qA_2 \otimes \qB_2)+(\qA_2 \otimes \qB_2\bqH_2) (\qC_1\otimes\qR_1) \qD_1^{-1} (\qA_2 \otimes \bqH_2^H \qB_2).
\end{align}


Similar to \eqref{eq:muC_2}, we also have
\begin{align}\label{eq:mu2C_2}
 \mu_0^{(r)}(\qC_0,\qC_1) = e^{-\calR_0^{(r)}(\qC_0,\qC_1)+ \calO(1)},
\end{align}
where $\calR_0^{(r)}(\qC_0,\qC_1)$ is the rate measure of $\mu_0^{(r)}(\qC_0,\qC_1)$ and is given by
\begin{align}\label{eq:RC2}
\calR_0^{(r)}(\qC_0,\qC_1) = & \max_{\tqC_0} \left\{ N \tr(\tqC_0\qC_0) - \log \Ex_{\qS} \left\{e^{ \tr \left(\tqC_0 \qS^H \qT_0 \qS \right)} \right\} \right\} \nonumber  \\
 &+\max_{\tqC_1} \left\{ N \tr(\tqC_1\qC_1) - \log \Ex_{\qS} \left\{e^{ \tr \left(\tqC_1 \qS^H \qT_1 \qS \right)} \right\} \right\},
\end{align}
with $\tqC_0 \in \bbC^{(r+1)\times(r+1)}$ and $\tqC_1 \in \bbC^{(r+1)\times(r+1)}$ are two symmetric matrices. Therefore, as $N\rightarrow\infty$, the integration over $\qC_0$ and $\qC_1$ can be performed via the saddle point method in \eqref{eq:ETr44}, we get
\begin{align}\label{eq:ETrmaxC2}
 &-\lim_{N\rightarrow\infty} \log \Ex_{\qS} \left\{e^{\tr\left( \qA_2 \qS^H \qH_1^H \bqH_2^H \qB_2 \bqH_2 \qH_1 \qS \right) + \tr \left(\tqC_2 \qS^H \qH_1^H \qT_2 \qH_1 \qS \right)} \right\} \nonumber  \\
 = & \min_{\qC_0,\qC_1} \left\{ - \calG_0^{(r)}(\qC_0,\qC_1) + \calR_0^{(r)}(\qC_0,\qC_1) \right\}.
\end{align}

By interchanging the two limits $N \rightarrow \infty$ and $r \rightarrow 0$ in \eqref{eq:F} and combining \eqref{eq:ETr}, \eqref{eq:ETrmaxC1}, and \eqref{eq:ETrmaxC2}, we obtain
\begin{align}
 F = & \lim_{r\rightarrow 0} \frac{\partial}{\partial r} \min_{\qC_2} \left\{ - \calG_2^{(r)}(\qC_2) + \calR_2^{(r)}(\qC_2) \right\} \nonumber  \\
   = & \lim_{r \rightarrow 0} \frac{\partial}{\partial r} \min_{\qC_2}\max_{\tqC_2} \bigg\{ K \log\left(1+ r\right) + \log\det\left(\qI + \qC_2 \qSigma \otimes \qR_2 \right) \nonumber  \\
       &+  L \tr(\tqC_2\qC_2) - \log \Ex_{\qS} \left\{e^{ \tr\left( \qA_2 \qS^H \qH_1^H \bqH_2^H \qB_2 \bqH_2 \qH_1 \qS \right) + \tr \left(\tqC_2 \qS^H \qH_1^H \qT_2 \qH_1 \qS \right)} \right\} \bigg\}  \nonumber  \\
   = & \lim_{r \rightarrow 0} \frac{\partial}{\partial r} \min_{\qC_0,\qC_1,\qC_2}\max_{\tqC_0,\tqC_1,\tqC_2} \bigg\{ K \log\left(1+ r\right) + \log\det\left(\qI + \qC_2 \qSigma \otimes \qR_2 \right) \nonumber  \\
       & + \log\det \left(\qI - \left(\qA_2 \otimes \bqH_2^H \qB_2 \bqH_2 +  \tqC_2 \otimes \qT_2 \right)\left( \qC_1 \otimes \qR_1 \right)\right) \nonumber  \\
       & + \log\det \left(\qI - \qD_2( \qC_0 \otimes \qR_0 )\right)  \nonumber  \\
       & + \log\det \left(\qI - \qA_0 \otimes \qB_0 - \tqC_0 \otimes \qT_0 - \tqC_1 \otimes \qT_1 \right)  \nonumber  \\
       &+ N \tr(\tqC_0\qC_0) + N \tr(\tqC_1\qC_1) + L \tr(\tqC_2\qC_2) \bigg\}.  \label{eq:F2}
\end{align}

\subsubsection{Replica Symmetry}
In order to obtain the saddle-points in \eqref{eq:F2}, we assume the known replica symmetry (RS), which the saddle-points are not affected by the dependence on the replica indices, rather than search over all possible $\qC_i$ and $\tqC_i$ for $i=0,1,2$. Therefore, we can write $\qC_i$ and $\tqC_i$ for $i=0,1,2$ as \cite{02TIT-Tanaka,07TIT-Wen}
\begin{equation}\label{eq:CandTC}
 \left\{
\begin{aligned}
    \qC_i & = (a_i-b_i)\qI_{(r+1)} + b_i \qone_{(r+1)} \qone_{(r+1)}^T,\\
    \tqC_i & = (\ta_i-\tb_i)\qI_{(r+1)} + \tb_i \qone_{(r+1)} \qone_{(r+1)}^T.
\end{aligned}
\right.
\end{equation}

With the RS in \eqref{eq:CandTC}, using Lemma \ref{Lemma:Eigen-deco simple}, we obtain that the eigenvalues of the matrix $\qC_i \qSigma$, $\qC_i$, and $\tqC_i$ for $i=0,1,2$ are given, respectively, by
\begin{subequations}\label{eq:Eigenvalue}
\begin{align}
 &\lambda_1(\qC_i \qSigma) = 0, ~~\lambda_a(\qC_i \qSigma) =  \frac{a_i-b_i}{\sigma^2},~~ a = 2, 3, \dots, r+1,  \\
 &\lambda_1(\qC_i) = a_i+rb_i, ~~\lambda_a(\qC_i) = a_i-b_i,~~ a = 2, 3, \dots, r+1, \\
 &\lambda_1(\tqC_i) = \ta_i+r\tb_i, ~~\lambda_a(\tqC_i) = \ta_i-\tb_i,~~ a = 2, 3, \dots, r+1.
\end{align}
\end{subequations}

Substituting \eqref{eq:Eigenvalue} into \eqref{eq:F2}, the second term of \eqref{eq:F2} can be rewritten as
\begin{align}\label{eq:second term F2}
 \log\det\left(\qI + \qC_2 \qSigma \otimes \qR_2 \right) = &r \log\det\left(\qI + \frac{a_2-b_2}{\sigma^2} \qR_2 \right).
\end{align}

Similarly, we get that the third, fourth, and fifth terms of \eqref{eq:F2} can be recast, respectively, as
\begin{align}\label{eq:third term F2}
 &\log\det \left(\qI - \left(\qA_2 \otimes \bqH_2^H \qB_2 \bqH_2 +  \tqC_2 \otimes \qT_2 \right)\left( \qC_1 \otimes \qR_1 \right)\right)  \nonumber  \\
 =& r\log\det\left(\qI + \left( \bqH_2^H \left((a_2-b_2)\qR_2+\sigma^2\qI \right)^{-1} \bqH_2 - (\ta_2-\tb_2)\qT_2 \right) (a_1-b_1)\qR_1 \right) \nonumber  \\
 &+ \log\det\left(\qI-(\ta_2+r\tb_2)\qT_2(a_1+rb_1)\qR_1\right),
\end{align}
\begin{align}\label{eq:fourth term F2}
 & \log\det \left(\qI - \qD_2( \qC_0 \otimes \qR_0 )\right) = r\log\det \left( \qI - (a_0-b_0) \qG_2 \qR_0 \right) +\log\det \qI,
\end{align}
and
\begin{align}\label{eq:fifth term F2}
 & \log\det \left(\qI - \qA_0 \otimes \qB_0 - \tqC_0 \otimes \qT_0 - \tqC_1 \otimes \qT_1 \right)  \nonumber  \\
 =& r \log\det \Big(\qI - \bqH_1^H \left( - \bqH_2^H \left( (a_2-b_2) \qR_2+\sigma^2\qI \right)^{-1} \bqH_2 + (\ta_2 - \tb_2) \qT_2 \right) \qG_1^{-1} \bqH_1 - \bqH_0^H \qG_2 \bqH_0 \nonumber  \\
 & +  \bqH_1^H \qG_1^{-1} \bqH_2^H \left( (a_2-b_2)\qR_2 + \sigma^2\qI \right)^{-1} \bqH_0 + \bqH_0^H \left( (a_2-b_2)\qR_2 + \sigma^2\qI \right)^{-1} \bqH_2 \qG_1^{-1} \bqH_1  \nonumber  \\
 & - \left(\bqH_0^H \qG_2 - \bqH_1^H \qG_1^{-1} \bqH_2^H \left( (a_2-b_2) \qR_2 + \sigma^2\qI \right)^{-1} \right) (a_0-b_0) \qR_0 (\qI-(a_0-b_0)\qG_2\qR_0)^{-1}  \nonumber  \\
 & ~~\times \left(\qG_2 \bqH_0 -\left( (a_2-b_2)\qR_2 + \sigma^2\qI \right)^{-1} \bqH_2 \qG_1^{-1} \bqH_1 \right) - (\ta_0-\tb_0) \qT_0 - (\ta_1-\tb_1) \qT_1 \Big) \nonumber  \\
 &+ \log\det \Big(\qI - (\ta_2 + r\tb_2)\bqH_1^H \qT_2 (\qI-(\ta_2 + r\tb_2)(a_1 + rb_1)\qT_2\qR_1)^{-1} \bqH_1  \nonumber  \\
  & ~~ - (\ta_0 + r\tb_0) \qT_0 - (\ta_1 + r\tb_1) \qT_1 \Big),
\end{align}
where
\begin{align}
 \qG_1=& \qI + \left( \bqH_2^H \left( (a_2-b_2) \qR_2+\sigma^2\qI \right)^{-1} \bqH_2 - (\ta_2-\tb_2)\qT_2 \right) (a_1-b_1)\qR_1, \nonumber  \\
 \qG_2=& - \left( (a_2-b_2) \qR_2+\sigma^2\qI \right)^{-1} \nonumber  \\
  &+(a_1-b_1) \left( (a_2-b_2) \qR_2+\sigma^2\qI \right)^{-1} \bqH_2 \qR_1 \qG_1^{-1} \bqH_2^H \left( (a_2-b_2) \qR_2+\sigma^2\qI \right)^{-1}. \nonumber
\end{align}

For the sixth, seventh, and eighth terms of \eqref{eq:F2}, we have
\begin{align}\label{eq:sixth term F2}
  N \tr(\tqC_0\qC_0) = N (r+1) \left(a_0\ta_0 + r b_0\tb_0\right),
\end{align}
\begin{align}\label{eq:seventh term F2}
  N \tr(\tqC_1\qC_1) = N (r+1) \left(a_1\ta_1 + r b_1\tb_1\right),
\end{align}
and
\begin{align}\label{eq:eighth term F2}
  L \tr(\tqC_2\qC_2) = L (r+1) \left(a_2\ta_2 + r b_2\tb_2\right),
\end{align}
respectively.

Substituting \eqref{eq:second term F2}-\eqref{eq:eighth term F2} into \eqref{eq:F2}, taking the derivative with respect to $r$, and equating the partial derivatives $\frac{\partial F}{\partial a_i}$, $\frac{\partial F}{\partial b_i}$, $\frac{\partial F}{\partial \ta_i}$, and $\frac{\partial F}{\partial \tb_i}~(i=0,1,2)$ to zeros, we can get the expressions of $\{a_i,b_i,\tb_i\}$ and $\ta_i=0$ if $r=0$ for $i=0,1,2$, and also have
\begin{align}
 F = & K + \log\det\left(\qI + \frac{a_2-b_2}{\sigma^2} \qR_2 \right)+ \log\det \left( \qI - (a_0-b_0) \qG_2 \qR_0 \right)\nonumber  \\
  & + \log\det\left(\qI + \left(\bqH_2^H \left((a_2-b_2)\qR_2+\sigma^2\qI \right)^{-1} \bqH_2 - (\ta_2-\tb_2)\qT_2 \right) (a_1-b_1)\qR_1 \right) \nonumber  \\
  & + \log\det \Bigg(\qI - \bqH_1^H \left( - \bqH_2^H \left( (a_2-b_2) \qR_2+\sigma^2\qI \right)^{-1} \bqH_2 + (\ta_2 - \tb_2) \qT_2 \right) \qG_1^{-1} \bqH_1 - \bqH_0^H \qG_2 \bqH_0 \nonumber  \\
 & +  \bqH_1^H \qG_1^{-1} \bqH_2^H \left( (a_2-b_2)\qR_2 + \sigma^2\qI \right)^{-1} \bqH_0 + \bqH_0^H \left( (a_2-b_2)\qR_2 + \sigma^2\qI \right)^{-1} \bqH_2 \qG_1^{-1} \bqH_1  \nonumber  \\
 & - \left(\bqH_0^H \qG_2 - \bqH_1^H \qG_1^{-1} \bqH_2^H \left( (a_2-b_2) \qR_2 + \sigma^2\qI \right)^{-1} \right) (a_0-b_0) \qR_0 (\qI-(a_0-b_0)\qG_2\qR_0)^{-1}  \nonumber  \\
 & ~~\times \left(\qG_2 \bqH_0 -\left( (a_2-b_2)\qR_2 + \sigma^2\qI \right)^{-1} \bqH_2 \qG_1^{-1} \bqH_1 \right) - (\ta_0-\tb_0) \qT_0 - (\ta_1-\tb_1) \qT_1 \Bigg) \nonumber  \\
  &  - N \left(a_0- b_0\right)\tb_0 - N \left(a_1- b_1\right)\tb_1 - L \left(a_2- b_2\right)\tb_2.  \label{eq:F3}
\end{align}

Defining $\te_i=\tb_i$ and $e_i=a_i-b_i$ for $i=0,1,2$ and substituting \eqref{eq:F3} into \eqref{eq:conditional mutual information}, we obtain this proposition.

\subsection{Taking the derivative of $\bR\left(\qQ,\qTheta\right)$ with respect to $\vartheta_l = e^{j\theta_l}$}\label{Appendix: Taking the derivative}

\begin{align}
\frac{\partial \bR}{\partial \vartheta_l} &= e_1 \tr\qF_1^{-1} \qF_3 \qR_1 + e_0 \tr \left( \qI_K + e_0 \qF_2 \qR_0 \right)^{-1} \qF_4 \qR_0 \nonumber  \\
  &+ \tr \qQ \left(\qI_N + \qF \qQ \right)^{-1} \Big(- e_1 \bqH_1^H \qF_1^{-1} \qF_3 \qR_1 \qF_1^{-1} \qTheta^H (\bqH_2^H \qPhi_2^{-1} \bqH_2 + \te_2\qT_2) \qTheta \bqH_1 + \bqH_1^H \qF_1^{-1} \qF_3 \bqH_1  \nonumber  \\
  &~~+ \bqH_0^H \qF_4 \bqH_0 - e_1\bqH_1^H \qF_1^{-1} \qF_3 \qR_1 \qF_1^{-1} \qTheta^H \bqH_2^H \qPhi_2^{-1} \bqH_0 - \vartheta_l^{-2}  \bqH_1^H \qF_1^{-1} \qE_{ll} \bqH_2^H \qPhi_2^{-1} \bqH_0  \nonumber  \\
     &~~- e_1 \bqH_0^H \qPhi_2^{-1} \bqH_2 \qTheta \qF_1^{-H} \qR_1 \qF_3 \qF_1^{-H} \bqH_1 + \bqH_0^H \qPhi_2^{-1} \bqH_2 \qE_{ll} \qF_1^{-H} \bqH_1 \nonumber  \\
     &~~ - \left( \bqH_0^H \qF_4 - e_1 \bqH_1^H \qF_1^{-1} \qF_3 \qR_1 \qF_1^{-1} \qTheta^H \bqH_2^H \qPhi_2^{-1} - \vartheta_l^{-2} \bqH_1^H \qF_1^{-1} \qE_{ll} \bqH_2^H \qPhi_2^{-1} \right) \nonumber  \\
       &~~~~\times e_0 \qR_0 (\qI+e_0\qF_2\qR_0)^{-1} \left(\qF_2 \bqH_0 + \qPhi_2^{-1} \bqH_2\qTheta \qF_1^{-H} \bqH_1 \right) \nonumber  \\
     &~~ - \left( \bqH_0^H \qF_2 + \bqH_1^H \qF_1^{-1} \qTheta^H \bqH_2^H \qPhi_2^{-1} \right) e_0 \qR_0 (\qI+e_0\qF_2\qR_0)^{-1}  \nonumber \\
       &~~~~\times \left(\qF_4 \bqH_0 - e_1 \qPhi_2^{-1} \bqH_2 \qTheta \qF_1^{-H} \qR_1 \qF_3 \qF_1^{-H} \bqH_1 + \qPhi_2^{-1} \bqH_2 \qE_{ll} \qF_1^{-H} \bqH_1 \right)  \nonumber  \\
     &~~+ e_0^2 \left( \bqH_0^H \qF_2 + \bqH_1^H \qF_1^{-1} \qTheta^H \bqH_2^H \qPhi_2^{-1} \right) \qR_0 (\qI+e_0\qF_2\qR_0)^{-1} \qF_4 \nonumber  \\
     &~~~~\times \qR_0 (\qI+e_0\qF_2\qR_0)^{-1} \left(\qF_2 \bqH_0 + \qPhi_2^{-1} \bqH_2\qTheta \qF_1^{-H} \bqH_1 \right) \Big), \label{eq:derivative barR}
\end{align}
where $\qE_{ll}$ denotes the all-zero $L \times L$ matrix except that the entry of the $l$-th row and $l$-th column is $1$, $\{ \qF, \qPhi_2, \qF_1, \qF_2\}$ are shown in \eqref{eq:qF123}, and $\{ \qF_3, \qF_4\}$ are given by
\begin{subequations}\label{eq:qF34}
\begin{align}
 \qF_3 = & \qTheta^H (\bqH_2^H \qPhi_2^{-1} \bqH_2 + \te_2\qT_2) \qE_{ll} - \vartheta_l^{-2} \qE_{ll} (\bqH_2^H \qPhi_2^{-1} \bqH_2 + \te_2\qT_2) \qTheta, \\
 \qF_4 = & e_1 \qPhi_2^{-1} \bqH_2 \left(\vartheta_l^{-2} \qTheta \qR_1 \qF_1^{-1} \qE_{ll} - \qE_{ll} \qR_1 \qF_1^{-1} \qTheta^H + e_1\qTheta \qR_1 \qF_1^{-1} \qF_3 \qR_1\qF_1^{-1} \qTheta^H \right) \bqH_2^H \qPhi_2^{-1}.
\end{align}
\end{subequations}

For the special case with the Rician factors of three channels $\kappa_t = \infty$ for $t =0,1,2$, i.e., the perfect CSIT is available at BS, we have
\begin{align}
\frac{\partial \bR}{\partial \vartheta_l} &= \frac{1}{\sigma^2} \tr \qQ \left(\qI_N + \qF \qQ \right)^{-1} \left( (\bqH_0 + \bqH_2 \qTheta  \bqH_1)^H \bqH_2 \qE_{ll} \bqH_1 - \vartheta_l^{-2} \bqH_1^H \qE_{ll} \bqH_2^H  (\bqH_0 +  \bqH_2 \qTheta \bqH_1 )\right).
\end{align}

When the direct link $\qH_0 = 0$, we have
\begin{align}
\frac{\partial \bR}{\partial \vartheta_l} &= \tr \qQ \left(\qI_N + \qF \qQ \right)^{-1} \left(- e_1 \bqH_1^H \qF_1^{-1} \qF_3 \qR_1 \qF_1^{-1} \qTheta^H (\bqH_2^H \qPhi_2^{-1} \bqH_2 + \te_2\qT_2) \qTheta \bqH_1 + \bqH_1^H \qF_1^{-1} \qF_3 \bqH_1 \right)\nonumber  \\
  &+ e_1 \tr\qF_1^{-1} \qF_3 \qR_1, \label{eq:derivative barR}
\end{align}

\subsection{Mathematical Tools}\label{Appendix: Mathematical Tools}

In this appendix, we provide some mathematical tools needed in this paper.

\begin{lemma}\label{Lemma:GIHST}
\cite[Lemma 1]{16TSP-Wen} (Gaussian Integral and Hubbard-Stratonovich Transformation) Let $\qz$ and $\qb$ be $N$-dimensional real vectors, and $\qA$ be an $N\times N$ positive definite matrix. Then,
\begin{align}
 \frac{1}{\pi^N} \int d\qz e^{-\qz^H\qA\qz + \qz^H\qb + \qb^H\qz} = \frac{1}{\det(\qA)} e^{\qb^H \qA^{-1}\qb}.
\end{align}
\end{lemma}

\begin{lemma}\label{Lemma:Eigen-deco simple}
\cite{12TWC-Wen} For a matrix $\qA = (a-b)\qI_{(r+1)} + b \qone_{(r+1)} \qone_{(r+1)}^T$, the eigen-decomposition of $\qA$ is given by
\begin{align}
  \qA = \qU\diag(a+rb,a-b,\ldots,a-b)\qU^H,
\end{align}
where $[\qU]_{n,m} = \frac{1}{\sqrt{r+1}} e^{-j\frac{2\pi}{r+1}(n-1)(m-1)}$ is the discrete Fourier transform matrix.
\end{lemma}



\begin{thebibliography}{10}
\providecommand{\url}[1]{#1}
\csname url@rmstyle\endcsname
\providecommand{\newblock}{\relax}
\providecommand{\bibinfo}[2]{#2}
\providecommand\BIBentrySTDinterwordspacing{\spaceskip=0pt\relax}
\providecommand\BIBentryALTinterwordstretchfactor{4}
\providecommand\BIBentryALTinterwordspacing{\spaceskip=\fontdimen2\font plus
\BIBentryALTinterwordstretchfactor\fontdimen3\font minus
  \fontdimen4\font\relax}
\providecommand\BIBforeignlanguage[2]{{%
\expandafter\ifx\csname l@#1\endcsname\relax
\typeout{** WARNING: IEEEtran.bst: No hyphenation pattern has been}%
\typeout{** loaded for the language `#1'. Using the pattern for}%
\typeout{** the default language instead.}%
\else
\language=\csname l@#1\endcsname
\fi
#2}}

\bibitem{20Network-Saad}
{W. Saad, M. Bennis, and M. Chen}, ``{A vision of 6G wireless systems:
  Applications, trends, technologies, and open research problems},'' \emph{IEEE
  Network}, vol.~34, no.~3, pp. 134--142, May/Jun. 2020.

\bibitem{14JSAC-Andrews}
{J. Andrews, S. Buzzi, W. Choi, S. Hanly, A. Lozano, A. C. K. Soong, and J.
  Zhang}, ``{What will 5G be?}'' \emph{IEEE J. Sel. Areas Commun.}, vol.~32,
  no.~6, pp. 1065--1082, Jun. 2014.

\bibitem{16JSAC-Buzzi}
{S. Buzzi, C. L. I, T. E. Klein, H. V. Poor, C. Yang, and A. Zappone}, ``{A
  survey of energy-efficient techniques for 5G networks and challenges
  ahead},'' \emph{IEEE J. Sel. Areas Commun.}, vol.~34, no.~4, pp. 697--709,
  Apr. 2016.

\bibitem{20ComM-Wu}
{Q. Wu and R. Zhang}, ``{Towards smart and reconfigurable environment:
  Intelligent reflecting surface aided wireless network},'' \emph{IEEE Commun.
  Magazine}, vol.~58, no.~1, pp. 106--112, Jan. 2020.

\bibitem{20WC-Huang}
{C. Huang, S. Hu, G. C. Alexandropoulos, A. Zappone, C. Yuen, R. Zhang, M. D.
  Renzo, and M. Debbah}, ``{Holographic MIMO surfaces for 6G wireless networks:
  Opportunities, challenges, and trends},'' \emph{IEEE Wireless Commun.},
  vol.~27, no.~5, pp. 118--125, Oct. 2020.

\bibitem{18TSP-Hu}
{S. Hu, F. Rusek, and O. Edfors}, ``{Beyond massive MIMO: The potential of data
  transmission with large intelligent surfaces},'' \emph{IEEE Trans. Sig.
  Proc.}, vol.~66, no.~10, pp. 2746--2758, May 2018.

\bibitem{19TWC-Huang}
{C. Huang, A. Zappone, G. C. Alexandropoulos, M. Debbah, and C. Yuen},
  ``{Reconfigurable intelligent surfaces for energy efficiency in wireless
  communication},'' \emph{IEEE Trans. Wireless Commun.}, vol.~18, no.~8, pp.
  4157--4170, Aug. 2019.

\bibitem{20TWC-Tang}
{W. Tang, M. Chen, X. Chen, J. Dai, Y. Han, M. D. Renzo, Y. Zeng, S. Jin, Q.
  Cheng, and T. Cui}, ``{Wireless communications with reconfigurable
  intelligent surface: Path loss modeling and experimental measurement},''
  \emph{IEEE Trans. Wireless Commun.}, vol.~20, no.~1, pp. 421--439, Jan. 2021.

\bibitem{20TWC-Nadeem}
{Q. U. A. Nadeem, A. Kammoun, A. Chaaban, M. Debbah, and M. S. Alouini},
  ``{Asymptotic max-min SINR analysis of reconfigurable intelligent surface
  assisted MISO systems},'' \emph{IEEE Trans. Wireless Commun.}, vol.~19,
  no.~12, pp. 7748--7764, Dec. 2020.

\bibitem{14LSA-CuiCoding}
{T. J. Cui, M. Q. Qi, X. Wan, J. Zhao, and Q. Cheng}, ``{Coding metamaterials,
  digital metamaterials and programmable metamaterials},'' \emph{Light Science
  \& Applications}, vol.~3, no.~10, p. e218, 2014.

\bibitem{20TWC-Jung}
{M. Jung, W. Saad, Y. Jang, G. Kong, and S. Choi}, ``{Performance analysis of
  large intelligent surfaces (LISs): Asymptotic data rate and channel hardening
  effects},'' \emph{IEEE Trans. Wireless Commun.}, vol.~19, no.~3, pp.
  2052--2065, Mar. 2020.

\bibitem{18GLOBECOM-Wu}
{Q. Wu and R. Zhang}, ``{Intelligent reflecting surface enhanced wireless
  network: Joint active and passive beamforming design},'' in \emph{Proc. IEEE
  Global Communi. Conf. (GLOBECOM)}, Abu Dhabi, United Arab Emirates, Dec.
  2018, pp. 1--6.

\bibitem{20WCL-Yan}
{W. Yan, X. Kuai, and X. Yuan}, ``{Passive beamforming and information transfer
  via large intelligent surface},'' \emph{IEEE Wireless Commun. Lett.}, vol.~9,
  no.~4, pp. 533--537, Apr. 2020.

\bibitem{20TCom-Abeywickrama}
{S. Abeywickrama, R. Zhang, Q. Wu, and C. Yuen}, ``{Intelligent reflecting
  surface: Practical phase shift model and beamforming optimization},''
  \emph{IEEE Trans. Commun.}, vol.~68, no.~9, pp. 5849--5863, Sep. 2020.

\bibitem{20TSP-Zhou}
{G. Zhou, C. Pan, H. Ren, K. Wang, and A. Nallanathan}, ``{Intelligent
  reflecting surface aided multigroup multicast MISO communication systems},''
  \emph{IEEE Trans. Sig. Proc.}, vol.~68, pp. 3236--3251, Apr. 2020.

\bibitem{20JSAC-Zhangshuowen}
{S. Zhang and R. Zhang}, ``{Capacity characterization for intelligent
  reflecting surface aided MIMO communication},'' \emph{IEEE J. Sel. Areas
  Commun.}, vol.~38, no.~8, pp. 1823--1838, Aug. 2020.

\bibitem{20TWC-PanCH}
{C. Pan, H. Ren, K. Wang, W. Xu, M. Elkashlan, A. Nallanathan, and L. Hanzo},
  ``{Multicell MIMO communications relying on intelligent reflecting
  surface},'' \emph{IEEE Trans. Wireless Commun.}, vol.~19, no.~8, pp.
  5218--5233, Aug. 2020.

\bibitem{20JSAC-Pan}
{C. Pan, H. Ren, K. Wang, M. Elkashlan, A. Nallanathan, J. Wang, and L. Hanzo},
  ``{Intelligent reflecting surface aided MIMO broadcasting for simultaneous
  wireless information and power transfer},'' \emph{IEEE J. Sel. Areas
  Commun.}, vol.~38, no.~8, pp. 1719--1734, Aug. 2020.

\bibitem{20TSP2-Zhou}
{G. Zhou, C. Pan, H. Ren, K. Wang, A. Nallanathan}, ``{A framework of robust
  transmission design for IRS-aided MISO communications with imperfect cascaded
  channels},'' \emph{IEEE Trans. Sig. Proc.}, vol.~68, pp. 5092--5106, Aug.
  2020.

\bibitem{20WCL-Zhou}
{G. Zhou, C. Pan, H. Ren, K. Wang, M. D. Renzo, and A. Nallanathan}, ``{Robust
  beamforming design for intelligent reflecting surface aided MISO
  communication systems},'' \emph{IEEE Wireless Commun. Lett.}, vol.~9, no.~10,
  pp. 1658--1662, Oct. 2020.

\bibitem{19TVT-Han}
{Y. Han, W. Tang, S. Jin, C.-K. Wen, and X. Ma}, ``{Large intelligent
  surface-assisted wireless communication exploiting statistical CSI},''
  \emph{IEEE Trans. Veh. Technol.}, vol.~68, no.~8, pp. 8238--8242, Aug. 2019.

\bibitem{ZhangICC20}
{J. Zhang, J. Liu, S. Ma, C.-K. Wen, and S. Jin}, ``{Transmitter design for
  large intelligent surface-assisted MIMO wireless communication with
  statistical CSI},'' in \emph{Proc. IEEE Int. Conf. on Commun. Workshop
  (ICC'20)}, Dublin, Ireland., Jun. 2020, pp. 1--5.

\bibitem{75JPF-Edwards}
{S. F. Edwards and P. W. Anderson}, ``{Theory of spin glasses},'' \emph{J.
  Physics F: Metal Physics}, vol.~5, pp. 965--974, 1975.

\bibitem{02TIT-Tanaka}
{T. Tanaka}, ``{A statistical-mechanics approach to large-system analysis of
  CDMA multiuser detectors},'' \emph{IEEE Trans. Inf. Theory}, vol.~48, no.~11,
  pp. 2888--2910, Nov. 2002.

\bibitem{03TIT-Moustakas}
{A. L. Moustakas, S. Simon, and A. M. Sengupta}, ``{MIMO capacity through
  correlated channels in the presence of correlated interfers and noise: A (not
  so) large N analysis},'' \emph{IEEE Trans. Inf. Theory}, vol.~49, no.~10, pp.
  2545--2561, Oct. 2003.

\bibitem{08JSAC-Muller}
{R. \"{M}uller, D. Guo, and A. Moustakas}, ``{Vector precoding for wireless
  MIMO systems and its replica analysis},'' \emph{IEEE J. Sel. Areas Commun.},
  vol.~26, no.~3, pp. 530--540, Apr. 2008.

\bibitem{08TIT-Taricco}
{G. Taricco}, ``{Asymptotic mutual information statistics of separately
  correlated Rician fading MIMO channels},'' \emph{IEEE Trans. Inf. Theory},
  vol.~54, no.~8, pp. 3490--3504, Nov. 2008.

\bibitem{10TWC-Wen}
{C. K. Wen, K. K. Wong, and J. C. Chen}, ``{Asymptotic mutual information for
  Rician MIMO-MA channels with arbitrary inputs: A replica analysis},''
  \emph{IEEE Trans. Commun.}, vol.~58, no.~10, pp. 2782--2788, Oct. 2010.

\bibitem{11TCom-Wen}
{C. K. Wen, K. K. Wong, and C. Ng}, ``{On the asymptotic properties of
  amplify-and-forward MIMO relay channels},'' \emph{IEEE Trans. Commun.},
  vol.~59, no.~2, pp. 590--602, Feb. 2011.

\bibitem{12TWC-Wen}
{C.-K. Wen, J.-C. Chen, and P. Ting}, ``{Robust transmitter design for
  amplify-and-forward MIMO relay systems exploiting only channel statistics},''
  \emph{IEEE Trans. Wireless Commun.}, vol.~11, no.~2, pp. 668--682, Feb. 2012.

\bibitem{13JSAC-ZhangJun}
{J. Zhang, C.-K. Wen, S. Jin, X. Q. Gao, and K.-K. Wong}, ``{On capacity of
  large-scale MIMO multiple access channels with distributed sets of correlated
  antennas},'' \emph{IEEE J. Sel. Areas Commun.}, vol.~31, no.~2, pp. 133--148,
  Feb. 2013.

\bibitem{10TIT-Dumont}
{J. Dumont, S. Lasaulce, W. Hachem, Ph. Loubaton, and J. Najim}, ``{On the
  capacity achieving covariance matrix for Rician MIMO channels: An asymptotic
  approach},'' \emph{IEEE Trans. Inf. Theory}, vol.~56, no.~3, pp. 1048--1069,
  Mar. 2010.

\bibitem{11Arxiv-Hoydis}
\BIBentryALTinterwordspacing
{J. Hoydis, R. Couillet, and M. Debbah}, ``{Iterative deterministic equivalents
  for the performance analysis of communication systems},'' 2011. [Online].
  Available: \url{http://arxiv.org/abs/1112.4167.}
\BIBentrySTDinterwordspacing

\bibitem{16TIFS-Zhang}
{J. Zhang, C. Yuen, C.-K. Wen, S. Jin, K.-K. Wong, and H. Zhu}, ``{Large system
  secrecy rate analysis for SWIPT MIMO wiretap channels},'' \emph{IEEE Trans.
  Inf. Forensics Security}, vol.~11, no.~1, pp. 74--85, Jan. 2016.

\bibitem{11TIT-Couillet}
{R. Couillet, M. Debbah, and J. W. Silverstein}, ``{A deterministic equivalent
  for the capacity analysis of correlated multi-user MIMO channels},''
  \emph{IEEE Trans. Inf. Theory}, vol.~57, no.~6, pp. 3493--3514, Jun. 2011.

\bibitem{20WCL-Bjornson}
{E. Bj\"{o}rnson, \"{O}. \"{O}zdogan, and E. G. Larsson}, ``{Intelligent
  reflecting surface versus decode-and-forward: How large surfaces are needed
  to beat relaying?}'' \emph{IEEE Wireless Commun. Lett.}, vol.~9, no.~2, pp.
  244--248, Feb. 2020.

\bibitem{21WCL-Nadeem}
{Q. U. A. Nadeem, A. Chaaban, and M. Debbah}, ``{Opportunistic beamforming
  using an intelligent reflecting surfaceWithout instantaneous CSI},''
  \emph{IEEE Wireless Commun. Lett.}, vol.~10, no.~1, pp. 146--150, Jan. 2021.

\bibitem{20PIMRC-Psomas}
{C. Psomas, I. Chrysovergis, and I. Krikidis}, ``{Random rotation-based
  low-complexity schemes for intelligent reflecting surfaces},'' in \emph{Proc.
  IEEE Int. Symposium on Personal, Indoor and Mobile Radio Communi. (PIMRC)},
  London, U.K., Aug. 2020, pp. 1--6.

\bibitem{07TIT-Wen}
{C.-K. Wen and K.-K. Wong}, ``{Asymptotic analysis of spatially correlated MIMO
  multiple-access channels with arbitrary signaling inputs for joint and
  separate decoding},'' \emph{IEEE Trans. Inf. Theory}, vol.~53, no.~1, pp.
  252--268, Jan. 2007.

\bibitem{16TSP-Wen}
{C.-K. Wen, J. Zhang, K.-K. Wong, J.-C. Chen, and C. Yuen}, ``{On sparse vector
  recovery performance in structurally orthogonal matrices via LASSO},''
  \emph{IEEE Trans. Sig. Proc.}, vol.~64, no.~17, pp. 4519--4533, Sep. 2016.

\end{thebibliography}
\end{document}